\documentclass[11pt]{article}

\usepackage{fullpage}
\usepackage{pdfpages}
\usepackage{graphicx} 
\usepackage{subcaption}
\usepackage[linktocpage=true]{hyperref}
\hypersetup{
  colorlinks=true,
  citecolor=blue,
  urlcolor=cyan,
  linkcolor=magenta,
}
\usepackage{format_arxiv}
\usepackage{notation}

\title{Game Dynamics and Equilibrium Computation \\
  in the Population Protocol Model}

\author{%
  Dan Alistarh\thanks{
    \texttt{dan.alistarh@ist.ac.at}
  }\\
  IST Austria \\
  \and
  Krishnendu Chatterjee\thanks{
    \texttt{krishnendu.chatterjee@ist.ac.at}
  }\\    
  IST Austria \\
  \and  
  Mehrdad Karrabi\thanks{
    \texttt{mehrdad.karrabi@ist.ac.at}
  }\\
  IST Austria \\
  \and
  John Lazarsfeld\thanks{
    \texttt{john.lazarsfeld@yale.edu}
  }\\
  Yale University \\
}

\date{}

\begin{document}

\maketitle
\thispagestyle{empty}

\begin{abstract}
We initiate the study of game dynamics in the 
\textit{population protocol} model: $n$ agents each
maintain  a current local strategy and interact in pairs
uniformly at random. Upon each interaction, the agents
play a two-person game and  receive a payoff from an
underlying utility function, and they can subsequently update 
their strategies according to a fixed local algorithm. 
In this setting, we ask how the distribution over agent
strategies evolves over a sequence of interactions, and
we introduce a new \textit{distributional equilibrium} 
concept to quantify the quality of such distributions. 
As an initial example, we study a class of
\textit{repeated prisoner's dilemma} games, and we consider 
a family of simple local update algorithms that yield 
non-trivial dynamics over the distribution of agent strategies.
We show that these dynamics are related to a new class
of \textit{high-dimensional Ehrenfest random walks}, and
we derive exact characterizations of their stationary 
distributions, bounds on their mixing times, and prove their 
convergence to approximate distributional equilibria. 
Our results highlight trade-offs between the local state space
of each agent, and the convergence rate and approximation 
factor of the underlying dynamics. Our approach opens the door 
towards the further characterization of equilibrium computation
for other classes of games and dynamics in the population setting.
\end{abstract}

\newpage
\setcounter{page}{1}

\section{Introduction}
\label{sec:intro}

The emergence of complex global behavior from the
 interactions of simple, computationally-limited
agents is a key topic of interest in distributed computing.
A standard setting is the \textit{population protocol} model,
in which a set of $n$ agents, modeled as simple,
anonymous state machines, interact randomly in pairs
with the goal of joint computation over the system's state. 
Since its introduction by
Angluin et al.~\cite{DBLP:journals/dc/AngluinADFP06}, 
this model has been used to characterize the evolution of
several families of dynamics for solving fundamental tasks
such as \textit{majority} \cite{DV12,PVV09,
  alistarh2015fast,Berenbrink16,GP16,AAG18, doty2022time} 
and \textit{leader election}~\cite{fischer2006self,
  AG15, DS18, GS18, berenbrink2020optimal}.
A key feature of this model is allowing to
characterize fine-grained notions of protocol convergence with
respect to population size, total number of pairwise
interactions (time), and available per-agent memory (space),
leading to interesting trade-offs between the space and
time complexity of simple local dynamics and their
complex global convergence behavior~\cite{DS18, alistarh2017time,
  doty2022time, berenbrink2020optimal}.

In this paper, we study population protocol dynamics in a
model where, upon each randomly scheduled pairwise interaction,
the two agents play a game, receive a payoff, and may
subsequently update their strategies.
Various multi-agent game settings have been studied
over the past decades, notably in they game theory and
evolutionary games literature, leading to work on proving
fast convergence guarantees, hardness results, and different
notions of global equilibrium for the
system~\cite{cesa2006prediction, 
  nowak2006evolutionary,nisan2007algorithmic,
  chen2009settling, daskalakis2009complexity,
  sandholm2010population,
  daskalakis2015approximate,
  golowich2020tight, hu2023best}.
Yet, most prior multi-agent game settings assume a sequence
of synchronous rounds, where the utility received by each 
individual agent depends on the actions of \textit{all other
agents in the population} in the most recent round
\cite{cesa2006prediction, cai2011minmax}.
In many real-world
settings, interactions are pair-wise, so an agent's utility may be \textit{local} and depend only
on its most recent interaction with a single other agent~\cite{bramoulle2014strategic}. 
These types of settings are well-captured
by the population protocol model, where agents interact randomly 
in pairs. However, studying the time and space complexity
of computing a global equilibrium in this model 
has yet to be considered.


We initiate such a study of equilibrium
computation in population protocols:
we introduce a new \textit{distributional equilibrium (DE)}
concept that captures a notion of
stability encoded by the distribution of individual
agent's pure strategies in the population. 
For a class of \textit{repeated prisoner's dilemma} games,
we then present a simple family of local update dynamics
that we prove converges to a stationary distribution
corresponding to an approximate DE.
In particular, these convergence results quantify
the non-trivial time, space, and approximation tradeoffs
that can arise in this multi-agent equilibrium computation
setting. For concreteness, we begin by introducing the
general problem setting in more detail. 

\subsection{Problem Setting}
\label{sec:intro:model}

\subsubsection{Multi-Agent Distributional Games}
\label{sec:intro:games}

We consider populations of $n$ agents, each of which
maintains a (pure) strategy $\S$ from a common
finite set $\calS$. At each time step $t$, two agents
$a$ and $b$ are sampled from the population uniformly
at random to interact and play a two-person game.
Letting $\S_1 \in \calS$ and $\S_2 \in \calS$ denote
the current strategies of $a$ and $b$, respectively,
agent $a$ receives a payoff of $u_1(\S_1, \S_2)$,
and agent $b$ receives a payoff of $u_2(\S_1, \S_2)$,
where $u_1$ and $u_2$ are fixed utility functions
for the underlying two-player game.

In this setting, a \textit{dynamics} $\calP$ is a local update
rule applied by agents $a$ and $b$ (following an interaction)
to determine their (possibly) new strategies, and this update
rule may depend on each agent's previous strategy and
the payoffs received from the interaction. Under
the randomness of the agent interactions, the dynamics induces
a sequence of distributions $\{\bfmu^t\}$,
where the coordinates of $\bfmu^t \in \Delta(\calS)$
specify the \textit{fraction of agents}\footnote{We write
  $\Delta(\calS)$ to denote the probability simplex
  over the finite set $\calS$.}
maintaining each strategy 
after the $t$'th total interaction.

Given the set of utility functions and a \textit{dynamics} $\calP$,
it is natural to ask whether (a) the distributions $\{\bfmu^t\}$
induced by $\calP$ \textit{converge} (for some appropriate
notion of convergence) to some distribution
$\bfmu \in \Delta(\calS)$ (and how quickly?), and  
(b) if so, whether $\bfmu$ can be characterized as
some global equilibrium (with respect to the utility functions
$u_1$ and $u_2$).
For this, we define the following \textit{distributional
  equilibrium} concept for a distribution $\bfmu \in \Delta(\calS)$:

\begin{definition}
  \label{def:de-gen}
  Let $\bfmu$ be a distribution over $\calS$.
  Then for $\eps > 0$, $\bfmu$ is an 
  \textit{$\eps$-approximate distributional equilibrium (DE)}
  if it satisfies the following:
  \begin{align}
    \E_{\S_1 \sim \bfmu, \;\S_2 \sim \bfmu}
    \big[u_{1}(\S_1, \S_2)\big]
    &\;\ge\;
      \max_{S' \in \calS}\;
      \E_{\S_2 \sim \bfmu}
      \big[u_{1}(\S', \S_2)\big]
      - \eps
    \label{eq:de:1} \\
    \text{and}\;\;
    \E_{\S_1 \sim \bfmu, \;\S_2 \sim \bfmu}
    \big[u_{2}(\S_1, \S_2)\big]
    &\;\ge\;
      \max_{S' \in \calS}\;
      \E_{\S_1 \sim \bfmu}
      \big[u_{2}(\S_1, \S')\big]
      - \eps  \;. 
    \label{eq:de:2}
  \end{align}
\end{definition}
An $\eps$-approximate DE can be viewed as
an approximate (symmetric) \textit{mixed Nash equilibrium}
(i.e., from classical game theory~\cite{nisan2007algorithmic}),
but where the ``mixed strategy'' $\bfmu \in \Delta(\calS)$
is \textit{induced} by the fraction of agents in the population
maintaining each pure strategy $\S \in \calS$.
In this view, the quantities
$\E_{\S_1 \sim \bfmu, \;\S_2 \sim \bfmu}
[u_{1}(\S_1, \S_2)]$ and
$\E_{\S_1 \sim \bfmu, \;\S_2 \sim \bfmu}
[u_{2}(\S_1, \S_2)]$ are the \textit{expected payoffs}
of two agents selected uniformly at random
(and whose current strategies are random
variables drawn independently from $\bfmu$),
and thus represent the agent payoffs from 
the ``average interaction'' in the population.
Then if $\bfmu$ is an $\eps$-approximate DE, the
inequalities~\eqref{eq:de:1} and~\eqref{eq:de:2} imply that
for either agent, a unilateral deviation to any
strategy $\S' \in \calS$ (while the other agent is still
selected uniformly at random with strategy $\S \sim \bfmu$)
can improve its expected payoff by at most $\eps$.

While the problem setting and distributional equilibrium
concept described above can be stated for
any general set of strategies $\calS$ and utilities
$u_1$ and $u_2$, as a starting point, we investigate
this concept for the special class of 
\textit{repeated prisoner's dilemma} games, 
and in particular, the subclass of 
\textit{repeated donation games}.
Repeated games have been studied for decades
in a variety of settings~\cite{owen2013game}, 
leading to both well-known ``folk theorems''~\cite{osborne1994course}, 
as well as to the study of evolutionary game dynamics
for simple strategies~\cite{nowak2006evolutionary}.
Moreover, repeated prisoner's dilemma is among the most classic 
repeated games, and it is frequently studied to model
the evolution of cooperation in populations
\cite{axelrod1981evolution, robert1984evolution}.
Thus to initiate the study of equilibrium computation 
in this population setting, we restrict our focus 
to this class of games, which we proceed to introduce 
in more detail:

\subsubsection{Repeated Prisoner Dilemma Games
  and Population Structure}
\label{sec:intro:RDG}

We consider \textit{repeated donation (RD) games}
(an important subclass of repeated prisoner's dilemma games)
and populations where agents have one of three
strategy types. In particular, we introduce the
game's reward structure and strategy types as follows:

\begin{itemize}[
  leftmargin=1em,
  itemsep=0.5em,
  ]
\item
  \textbf{Reward Structure}:
  In RD games, two players begin
  by playing a single round of \textit{prisoner's dilemma}
  (PD). At the end of each round, an \textit{additional}
  round may be played with (independent) probability $\delta$;
  otherwise, the game terminates.    
  We call $\delta$ the
  \textit{continuation} or \textit{restart probability}.
  In a single round of PD,
  each player simultaneously chooses
  to \textit{cooperate} ($C$) or \textit{defect} ($D$),
  and the eponymous \textit{dilemma} is that
  each player's payoff-maximizing decision is to defect,
  despite the fact that mutual cooperation leads
  to a higher payoff.
  In particular, we consider payoffs whose
  reward structure are \textit{donation games},
  which are the most important class
  of PD rewards~\cite{hilbe2013evolution,
    marshall2009donation, stewart2013extortion}.
  These rewards are captured via a
  \textit{reward vector} $\v := [b-c, -c, b, 0]^\top$
  over the four game states $\calA := \{CC, CD, DC, DD\}$,
  which are are defined by the ordered actions of the
  first and second players. Here, the
  entries in $\v$ satisfy $b > c \ge 0$
  and specify the reward of the first player.
  Each player's \textit{total reward}
  is then the sum of its payoffs
  over the individual rounds of the game.
\item
  \textbf{Strategy Types and $(\alpha, \beta, \gamma)$
    Populations}:
  We assume that each of the $n$ agents in the population
  belongs to one of three subpopulations that correspond to
  distinct RD strategy types:
  an $\alpha$ fraction have strategy type
  $\AC$ (Always-Cooperate),
  a $\beta$ fraction have strategy type $\AD$ (Always-Defect),
  and a $\gamma$ fraction have strategy type $\gtft$
  (Generous-Tit-For-Tat),
  where $\alpha + \beta + \gamma = 1$ (and thus
  we refer to such populations as $(\alpha, \beta, \gamma)$
  populations). 
  An agent's strategy determines
  its action ($C$ or $D$) in each round of the repeated game
  and may depend on its opponent's actions from prior rounds.
  Specifically:
  \begin{itemize}[
    label=-,
    itemsep=0pt,
    ]
  \item
    \textit{$\AC$ strategy}: play $C$ at each round.
  \item
    \textit{$\AD$ strategy}: play $D$ at each round.
  \item
    \textit{$\gtft$ strategy}:
    given a \textit{generosity parameter} $g \in [0,1]$,
    play $C$ with initial cooperation probability
    $s_1 \in [0, 1]$, and $D$ with probability (w.p.)
    $1-s_1$. In round $r+1$, play the
    opponent's action from round $r$ w.p. $(1-g)$,
    and play $C$ w.p. $g$. 
  \end{itemize}
  Moreover, we assume the strategy of $\AC$ and $\AD$ agents 
  \textit{always remains fixed}, and thus
  we consider update dynamics $\calP$ that are only followed by
  agents with strategy $\gtft$.
  In particular, we assume each $\gtft$ agent maintains a
  strategy from the set $\calG = \{g_1, \dots, g_k\}$
  of $k \ge 2$ generosity parameter values, each of which
  are defined as follows:
  given a \textit{maximum generosity parameter}\footnote{
    Assuming such bounds on the generosity probability
    in the $\gtft$ strategy are standard in RPD settings
    \cite{nowak2006evolutionary}.
  } $\hatg \le 1$, for each $i \in [k]$
  (and by a slight abuse of notation),
  the strategy $g_i$ is the $\gtft$ strategy with
  generosity parameter $g_i := \hatg \cdot \big(\frac{i-1}{k-1}\big)$.
  Here, $\calG$ can be viewed as a discretization of the
  continuous space $[0, \hatg]$ into $k$ equidistant
  generosity parameter values.
\item  
  \textbf{Expected Payoff Functions}:
  Given the single RD game reward structure and the
  strategy types of the $(\alpha, \beta, \gamma)$ population,
  we can define an agent's utility function in the
  context of the more general
  multi-agent distributional game setting
  from Section~\ref{sec:intro:games}. 
  Specifically, we assume that when agents $a$ and $b$
  interact, each agent's utility is its 
  \textit{expected payoff} in an RD game
  (over the randomness in each agent's strategy
  and the repeated rounds).   
  For a pair of strategies
  $\S_1, \S_2 \in \calS := \{\AC, \AD, g_1, \dots, g_k\}$,
  we let $f(\S_1,\S_2)$ denote this expected payoff
  for an agent with strategy $\S_1$ against
  an opponent with strategy $\S_2$.
  By the symmetry of the single-round RD rewards,
  it follows that $f(\S_2, \S_1)$ is the expected
  payoff for the agent with strategy $\S_2$. 
\end{itemize}

\paragraph{Discussion}
We remark that the focus on the three strategy types
above is two-fold: 
first, a classical strategy in RPD is
tit-for-tat ($\tft$)~\cite{axelrod1981evolution},
which always repeats the opponent's previous action
in the next round. It can be shown that the $\tft$
strategy leads to the emergence of cooperation under
suitable parameter values~\cite{axelrod1981evolution,
  nowak2006evolutionary}, however its main drawback is
lack of robustness:
even in the two-player sequential setting, in the presence of 
noise or errors where a cooperative action may be replaced by defection,
a single error makes two $\tft$ players alternate between $C$ and $D$,
and after two errors both players will choose to defect forever. 
The key mechanism to deal with such errors is the introduction
of \emph{generosity}~\cite{stewart2013extortion, martinez2012generosity, molander1985optimal, nowak1992tit}, which motivates the $\gtft$
class of strategies defined above and is the focus of this paper.
While the analysis of other reactive strategies in RD
games extends beyond $\gtft$~\cite{nowak2006evolutionary},
it is natural to begin the study of equilibrium computation 
with the simpler strategy sets described above. 
Moreover, several recent works on population
protocols have considered populations containing 
a subset of agents whose states always remain fixed~\cite{
  alistarh2017time,DBLP:conf/stoc/DudekK18,amir2020approximate,
  alistarh2021comparison,berenbrink2021loosely}.
An $(\alpha, \beta, \gamma)$ population
can be viewed as an example of this setting, where the
$\alpha$ and $\beta$ fractions of $\AC$ and $\AD$ agents
remain invariant under any update dynamics.
 
To that end, given the RD and $(\alpha, \beta, \gamma)$
population structure described here,
we can now more precisely describe the distributional equilibrium
concept and the main algorithmic questions of interests
for this setting. 

\subsubsection{Distributional Equilibrium
  in RD games for $(\alpha, \beta, \gamma)$ Populations}
\label{sec:intro:de-rd}   

In this setting, we consider update dynamics for
agents with $\gtft$ strategies, and we are interested
in the resulting distributions $\{\bfmu^t\}$ over $\calG$
that specify the fraction of $\gtft$ agents with
strategy parameter $g_i$ for each $i \in [k]$.

For a fixed $k \ge 2$ and some distribution 
$\bfmu \in \Delta(\calG)$, and given 
the $(\alpha, \beta, \gamma)$ population parameters,
$\bfmu$ also induces the $(k+2)$-part distribution
$\hatbfmu \in \Delta(\calS)$ over the full strategy set
$\calS$, where
$\calS := \{\AC, \AD, g_1, \dots, g_k\}$.
Specifically, given $\bfmu$, define:
\begin{equation}
  \hatbfmu(\AC) = \alpha,
  \;\;\;
  \hatbfmu(\AD) = \beta,
  \;\;\;\text{and}\;\;
  \hatbfmu(i) = \gamma \cdot \bfmu(i)
  \;\;\;\text{for $i \in [k]$}.
\end{equation}
Then the $\eps$-approximate distributional equilibrium
concept of Definition~\ref{def:de-gen} extends to
the $(\alpha, \beta, \gamma)$ population setting as follows:
\begin{definition}
  \label{def:de}
  In the RD setting with an $(\alpha, \beta, \gamma)$ population,
  and for $\eps > 0$, we call
  $(\bfpi, \hatbfpi) \in \Delta(\calG) \times \Delta(\calS)$ an
  \textit{$\eps$-approximate distributional equilibrium (DE)}
  if it satisfies the following property:
  \begin{equation}
    \E_{g \sim \bfpi, \;\S \sim \hatbfpi}
    \big[f(g, \S)\big]
    \;\ge\;
    \max_{g' \in \calG}\;
    \E_{\S \sim \hatbfpi}
    \big[f(g', \S)\big]
    - \eps \;.
    \label{eq:de}
  \end{equation}
  Given the one-to-one correspondence between
  $\bfpi$ and $\hatbfpi$, if expression~\eqref{eq:de} holds,
  we will simply say that $\bfpi$ is an $\eps$-approximate
  distributional equilibrium.
\end{definition}

Here, the equilibrium definition
can be viewed as being restricted to agent interactions
where at least one agent has a $\gtft$ strategy
(this is motivated by the fact that
$\AD$ and $\AC$ agents never change their strategies).
In other words,
$\E_{g \sim \bfpi, \;\S \sim \hatbfpi}
[f(g, \S)]$ is the expected RD game payoff of a randomly selected
$\gtft$ agent (under the distribution $\bfmu$ over $\calG$) 
playing against a randomly selected opponent with any strategy
(under the distribution $\hatbfmu$ over $\calS$ that is 
induced by $\bfmu$)
If $\bfmu$ satisfies expression~\eqref{eq:de},
then the first randomly selected $\gtft$ agent can improve
its expected payoff by no more than $\eps$ when 
unilaterally deviating to a different $\gtft$ strategy parameter
(while still playing against an opponent whose strategy
is drawn from $\hatbfmu$).
By symmetry of the expected RD payoff functions $f$,
expression~\eqref{eq:de} also captures this same property
if the second agent in the interaction is conditioned
to be a $\gtft$ agent.

With this $\eps$-approximate DE definition in hand,
we can more precisely state two natural questions
of interest in this setting:
\begin{enumerate}[
  leftmargin=1em,
  rightmargin=0.5em,
  itemsep=0.5em,
  label=-
  ]
\item
  \textit{(Q1): Is there a local dynamics $\calP$
  for updating $\gtft$ agents' strategies that converges
  to a distribution $\bfmu \in \Delta(\calG)$ that
  is an $\eps$-approximate DE, and for what approximation parameter $\eps$?}
\item
  \textit{(Q2): If so, how many interactions
  are needed for $\calP$ to 
  converge to $\bfmu$?}
\end{enumerate}

\subsection{Our Contributions}

We introduce a family of update dynamics
for the setting above that converges to an approximate DE,
and we obtain quantitative bounds on its 
convergence rate and approximation factor:

\begin{itemize}[
  leftmargin=1em,
  itemsep=0.5em,
  ]
\item
  \textbf{$k$-IGT Dynamics} (Definition~\ref{def:kigt}):
  We define a family of natural
  \textit{incremental generosity tuning (IGT)}
  dynamics for updating the generosity parameters
  of $\gtft$ agents among the set $\calG = \{g_1, \dots, g_k\}$
  for each $k \ge 2$.
  In each instantiation, after a $\gtft$ agent
  interacts with a second agent from the population,
  it increments or decrements its generosity parameter
  to the next highest or smallest parameter value in $\calG$,
  depending on the strategy type of its opponent.
  We show this update rule can be viewed as an
  \textit{introspection dynamics with local search}
  from evolutionary games~\cite{nowak2006evolutionary}, 
  and we abbreviate the $k$-th dynamics by $k$-IGT.
\item
  \textbf{Stationary properties of $k$-IGT}
  (Theorem~\ref{thm:kigt-main}):
  We derive exact characterizations of the
  stationary distribution of the $k$-IGT dynamics over
  the set of generosity parameters $\calG$,
  and we derive bounds on its mixing time:
  \begin{enumerate}[
    label={(\roman*)},
    itemsep=0pt,
    leftmargin=2em,
    ]
  \item
    the stationary distribution of the $k$-IGT dynamics
    is \textit{multinomial} with parameters $\gamma \cdot n$
    and $(p_1, \dots, p_k)$,
    where each $p_j \propto (1/\beta - 1)^{j-1}$.
  \item
    when $\beta \neq 1/2$,
    the mixing time of the dynamics is at most 
    $O(kn \log n)$ total interactions, and when $\beta = 1/2$,
    the dependence on $k$ changes to $k^2$. 
    Moreover, the mixing time is lower bounded 
    by at least $\Omega(kn)$ total interactions.
  \end{enumerate}
\item
  \textbf{Convergence to an $\eps$-approximate DE}
  (Theorem~\ref{thm:kigt-eq-converge}):
  We prove that under suitable population and game setting
  regimes (when the maximum generosity parameter $\hatg$
  is bounded with respect to $b, c, \delta$ and
  the population parameters $(\alpha, \beta, \gamma)$),
  the \textit{mean} of the stationary
  distribution $\bfmu$ for the $k$-IGT dynamics
  is an $\eps$-approximate DE
  (Definition~\ref{def:de}) for $\eps = O(1/k)$. 
\end{itemize}

Our results highlight interesting
tradeoffs in the multi-agent computation of an
$\eps$-approximate DE in this setting: for increasing
$k$, meaning larger generosity parameter spaces $\calG$
and larger local state requirements for each $\gtft$ agent,
the time (number of total interactions) to converge
to the stationary distribution of the $k$-IGT dynamics
grows linearly with $k$, but it results in a 
better equilibrium approximation factor $\eps$ decaying 
at a rate of $1/k$.
In general, these are the first such results that
quantify the convergence properties
(both in mixing time, and in equilibrium approximation)
of game dynamics in multi-agent settings with
random pairwise interactions. 

The main technical tool to achieve these results is to 
relate the evolution of the $k$-IGT dynamics to a new family of 
\textit{high-dimensional, weighted Ehrenfest random walks},
which generalize the classic two-dimensional
Ehrenfest urn process~\cite{ehrenfest1907zwei}.
The stationary and mixing results we prove for these 
new Ehrenfest processes may be of independent interest. 
Our results also leave open the more general question of
designing dynamics for other classes of games,
and to quantify their resulting convergence properties
to approximate distributional equilibria.

\paragraph{Structure of Paper}
The remainder of the paper is structured as follows:
in Section~\ref{sec:intro:related}, we discuss
several lines of related work in more detail.
In Section~\ref{sec:tech-overview}, we introduce some
additional preliminaries and provide
a technical overview of our families of dynamics and 
convergence results. 
The proofs of most results are deferred
to the appendix.

\subsection{Related Work}
\label{sec:intro:related}

\paragraph{Population protocol dynamics}
The population model was originally introduced by
Angluin et al.~\cite{DBLP:journals/dc/AngluinADFP06} to
model computation in populations of passively mobile
agents (such as sensor networks or animal populations),
and has since found several other applications,
from chemical reaction networks~\cite{Doty14, LCK16, CCDS15}
to computing via synthetic DNA strands~\cite{CDSPCSS13}.
On the theoretical side, an impressive amount of effort has
been invested in understanding the computational power of the
model, e.g.~\cite{DBLP:journals/dc/AngluinADFP06, AAE08}, 
on analyzing fundamental dynamics such as rumor spreading and averaging~\cite{giakkoupis2016asynchrony, becchetti2018average},
and on the complexity of core algorithmic tasks such
as majority (consensus)~\cite{DV12,PVV09, 
  becchetti2014plurality, alistarh2015fast,
  Berenbrink16,GP16,AAG18, becchetti2020consensus, 
  doty2022time} 
and leader election~\cite{fischer2006self,
  AG15, DS18, GS18, berenbrink2020optimal}
in this model. 
The latter direction has recently lead to tight bounds on 
the space and task complexity of these tasks~\cite{AAG18, 
berenbrink2020optimal, doty2022time}.
In this context, our contributions are to design and analyze
a novel class of repeated game dynamics that lead to
interesting time and space tradeoffs.

\paragraph{Evolutionary game dynamics} 
There is a huge literature on evolutionary game dynamics,
and we briefly mention some key results and the relationship 
to our work. The first approach is to consider evolutionary 
dynamics in an infinite population with the aid of 
differential equations (aka, the replicator dynamics)
\cite{smith1982evolution, nowak2006evolutionary},
and the goal is to study the existence
and stability of equilibrium points.  
The second approach is to consider evolutionary dynamics
in finite populations with a specific class of 
strategies e.g., reactive strategies or memory-1 strategies.
In these approaches, the class of strategies
is uncountable and simulation results suggest which strategies 
successfully evolve in the simulation
of evolutionary dynamics~\cite{nowak1992tit, nowak1993strategy}.
The third approach is to consider evolutionary dynamics on networks 
but with only two strategy types
($\AC$ and $\AD$)~\cite{lieberman2005evolutionary, 
  allen2017evolutionary}.
In contrast to the present paper, none of these works
focus on quantitative aspects related to the mixing time
(or convergence time) to the stationary behavior.

\paragraph{Complexity of equilibrium computation}
Additionally, in both classical and evolutionary game settings,
much attention has been given to understanding 
the computational complexity of computing an equilibrium:
for example, it is known in general that computing 
a Nash equilibrium in a general-sum two-player game 
is PPAD-hard~\cite{chen2009settling, daskalakis2009complexity}.
Moreover, even determining the existence of evolutionary 
stable strategies in two-player strategic form 
games is NP-hard~\cite{etessami2008computational}, and 
computing the fixation probability in evolutionary 
games with two strategies on graphs is 
PSPACE-complete~\cite{ibsen2015computational}.
In light of this, much progress has also been made
on designing dynamics that provably converge
for either (a) special subclasses of games
(such as zero-sum games \cite{nash1951non})
or (b) to approximate or weaker equilibrium concepts 
(such as correlated or coarse correlated equilibria
\cite{aumann1987correlated, roughgarden2016twenty, 
  golowich2020tight}).


\paragraph{Other multi-player game settings}
As mentioned earlier in the introduction, 
there is a large body of literature
on other multi-player game settings. 
In the most basic, multi-agent normal-form
game setting~\cite{cesa2006prediction, nisan2007algorithmic},
$n$ players simultaneously
choose actions at each round, and each receives
a payoff according to a reward function
that depends on the actions of \emph{all other players}.
In this setting and its closely related variants, 
extensive work has been devoted to designing
local strategies that provably converge to
an equilibrium, and to determine their corresponding 
rates of convergence (e.g.,~\cite{cai2011minmax, 
  daskalakis2015approximate,   
  lin2020finite, golowich2020tight, hu2023best}).
This is in contrast to the setting of the present work, 
where only a \textit{single random pair} of agents
interacts at each step, and thus these results
are not directly comparable.

An orthogonal line of work to ours previously
investigated game theoretic aspects of population
protocols~\cite{bournez2009playing, bournez2013population},
but the focus in these works is on understanding
the \textit{computational power} of interaction rules
that correspond to symmetric games. 
Additionally, some prior work~\cite{dyer2002convergence}
studied certain random walk processes 
that model the evolution of cooperation in 
repeated prisoner's dilemma games in populations of interacting
agents. However, the focus in that work is restricted
to bounding the stabilization time of the process, 
and not on its characterization within the framework
of equilibrium computation. 


\section{Technical Overview}
\label{sec:tech-overview}

\subsection{Preliminaries}
\label{sec:tech-overview:prelims}

\paragraph{Notation}
We use the shorthand notation $[k] = \{1, \dots, k\}$
for any $k \ge 0$, and we define the set
$\Delta^m_k :=
\big\{ (x_1,\dots, x_k)  \in \mathbb{N}^k \;:\;
\sum_{j \in [k]} x_j = m \big\}$.
For non-negative integers $x$, $a$, and $b$,
we write $[x]_a^b$ to denote $[x]$ truncated to the range $[a, b]$.
For readability, when $a$ and $b$ are clear from context,
we will (by slight abuse of notation) simply write $[x]$. 

\paragraph{Markov chains}
We consider discrete-time Markov chains $\{\x^t\}$
over a discrete space $\Omega$,
with transition matrix $\P: \Omega \times \Omega \to [0, 1]$.
Recall that $\bfpi: \Omega \to [0, 1]$ is a
\textit{stationary distribution} of $\{\x^t\}$ if
$\bfpi \P = \bfpi$  (where we interpret the
probability mass function (PMF) of $\bfpi$ as a row vector).
Recall also that any distribution $\bfnu: \Omega \to [0, 1]$
satisfying the \textit{detailed balance equations}
$\bfnu(\x) \P(\x,\y) = \bfnu(\y) \P(\y, \x)$ for all $\x, \y \in \Omega$
is a stationary distribution for the process.
Starting from $\x^1 := \x$ for any $\x \in \Omega$,
we let $\P^t(\x)$ denote the the distribution of $\x^t$
(i.e., of the process after $t$ steps), and
we write $d(t) := \max_{\x \in \Omega} \| \P^t(\x) - \bfpi \|_{\tv}$
to denote the \textit{distance to stationarity} (in total variation)
of the process after $t$ steps, maximized over all initial states.
Then we define the \textit{mixing time} $\tmix$ of $\{\x^t\}$ 
as $\tmix := \min \{ t \ge 0 \;:\; d(t) \le 1/4\}$. 
We refer the reader to the text of
Levin and Peres~\cite{levin2017markov}
for more background and preliminaries on Markov chains
and mixing times. 

\paragraph{Multinomial distributions}
We recall basic facts about multinomial distributions.
For $m \ge 1$, $k \ge 2$, and a sequence
$(p_1, \dots, p_k)$ such that $p_1 + \dots + p_k = 1$,
a distribution $\bfnu$ is
\textit{multinomial with parameters $m$ and $(p_1, \dots p_k)$}
if the PMF of $\bfnu$ is given by
$
  \bfnu(\x)
  \;=\;
  p_1^{x_1} \cdot p_2^{x_2} \dots p_k^{x_k}
  \cdot \binom{m}{x_1, \dots, x_k}
$
for all $\x = (x_1, \dots, x_k) \in \Delta^m_k$,
where the multinomial coefficient is defined as
$
  \binom{m}{x_1, \dots, x_k}
  \;=\;
  \frac{m!}{x_1! \cdot x_2! \dots x_k!}
$
.
Writing $\bfnu = (\nu_1, \dots, \nu_k)$, it is known
that $\E[\nu_j] = m \cdot p_j$ for all $j \in [k]$.
When $k=2$, then $\bfnu$ is a \textit{binomial distribution},
and we can simply say that $\bfnu$ is 
\textit{binomial with parameters $m$ and $p_1$}.

\subsection{$k$-IGT Dynamics}

We begin by formally introducing our family of local update dynamics
for the problem setting described in Section~\ref{sec:intro:model}.
Fixing $k \ge 2$, recall that we assume each $\gtft$ agent
maintains a generosity parameter $g \in \calG = \{g_1, \dots, g_k\}$
at each time step. The $k$-IGT dynamics then follows two
transition types:
(a) after a $\gtft$ agent $u$ interacts with an $\AC$
agent or a second $\gtft$ agent, $u$ increases its generosity
parameter to the next largest value in $\calG$,
and (b) after a $\gtft$ agent $u$ interacts with an $\AD$
agent, $u$ decreases its generosity parameter to the
next smallest value in $\calG$. Defined formally: 

\begin{restatable}[Incremental-Generosity-Tuning (IGT) Dynamics]
  {definition}{kigtdef}
  \label{def:kigt}
  Consider an $(\alpha, \beta, \gamma)$ population
  and an RD game setting with maximum generosity parameter
  $\widehat g$. 
  For any $k \ge 2$, define the set of $k$ generosity parameters
  $\calG = \{g_1, \dots, g_k\}$
  where each $g_j = \hg \cdot \big(\frac{j-1}{k-1} \big)$. 
  Randomly initialize the parameter of
  every $\gtft$ agent to some $g \in \calG$.
  Then the $k$-IGT dynamics is the population
  protocol that evolves for all $j \in \{1, \dots, k\}$
  according to the following
  transitions over the strategy types of interacting nodes:
  \begin{equation*}
    \begin{aligned}
    \text{(i)}\;\;
    &g_j \;+\;  \AC
      \;\longrightarrow\;
      \inc(g_j) \;+\; \AC \\
    \text{(ii)}\;\;
    &g_j \;+\;  g_i \;\;
      \;\longrightarrow\;
      \inc(g_j) \;+\; g_i \;,
      \;\;\;\;\;\text{for all $i \in[k]$}\\
    \text{(iii)}\;\;
    &g_j \;+\;  \AD
      \;\longrightarrow\;
      \dec(g_j) \;+\; \AD  \;\;,
    \end{aligned}
  \end{equation*}
  where $\inc(g_j) := g_{\min\{j+1, k\}}$
  and $\dec(g_j) := g_{\max\{j-1, 1\}}$. 
\end{restatable}

\begin{figure}[tb!]
  \centering
  \vspace*{-1em}
  \includegraphics[scale=0.65]{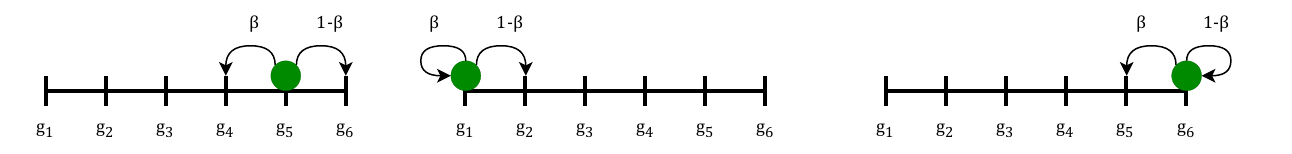}
  \caption{\small
    When $k=6$, three examples showing how the parameter value
    of a $\gtft$ agent is probabilistically updated under
    the $k$-IGT dynamics. Note that conditioned on a $\gtft$ agent $u$
    being sampled as the first agent, $u$ increments its
    parameter value with probability $(1-\beta)$ 
    and decrements its parameter value 
    with probability $\beta$ (where in both cases,
    the values are truncated to the range $[0, \hg]$).
  }
  \label{fig:examples}
\end{figure}
Figure~\ref{fig:examples} shows an example
of how the parameter value of a $\gtft$ agent is updated
depending on the strategy type of its interaction partner.\footnote{
  We refer to the first agent in such interactions as
  the \textit{initiator}, and in this setting we
  assume only that the initiator
  ever updates its strategy following an interaction.
  This type of \textit{one-way} protocol is a standard
  modeling assumption in the population
  protocol literature, e.g.  
  \cite{DBLP:journals/dc/AngluinADFP06, AngluinAE08, AAE08,
    becchetti2014plurality, becchetti2020consensus}.
}
Note that while the transition rules of the k-IGT dynamics 
are defined with respect to the \textit{strategies} of the two agents,
these transitions could alternatively be defined with respect to
the \textit{observed game actions}.
We remark that for sufficiently large $\delta$
the resulting transition rules will be essentially the same
as in Definition~\ref{def:kigt},
as in this case each agent can infer the strategy type
of its opponent with high probability. 
It follows that the resulting dynamics will be
essentially the same as to those induced
by Definition~\ref{def:kigt}, up to some small approximation error.
Thus for simplicity, the transitions of the $k$-IGT dynamics
are defined with respect to the strategy types of the
interacting agents. 

\paragraph{Bridging $k$-IGT and introspection dynamics}
Under mild constraints on the reward vector $\v$
and maximum generosity parameter $\hg$,
we show that the transition rules 
of the $k$-IGT dynamics are 
\textit{locally optimal} in the following sense:
under any transition rule of Definition~\ref{def:kigt},
the expected payoff $f(g, \S)$
will \textit{never decrease} had the $\gtft$ agent
used the \textit{updated} generosity parameter value
(specified by the transition rules) against
its previous opponent with strategy $\S$. 
In particular, this bridges the relationship
between the $k$-IGT dynamics and
the classic concept of \textit{introspection dynamics
  with local search}
from evolutionary games~\cite{hofbauer1998evolutionary},
where an agent explores the local neighborhood
of its strategy space to
adopt a new strategy that would have performed better.
Formally, in
Appendix~\ref{app:gt} 
we prove the following
proposition, which relies on (i) deriving
exact expressions for the expected payoffs $f(g, \cdot)$,
and (ii) differentiating these payoff functions with
respect to $g$:

\begin{restatable}{proposition}{params}
  \label{prop:params}
  Consider an RD game setting consisting of 
  (a) an initial cooperation probability $s_1 \in [0, 1)$
  (b) a restart probability 
  $\delta > \tfrac{c}{b}$ and
  (c) a maximum generosity parameter
  $\widehat g < 1 - \frac{c}{\delta b}$.
  Then for all $g, g' \in  [0, \widehat g]$
  such that $g < g'$, 
  the following three statements hold:
  \begin{enumerate}[
    label={(\roman*)},
    leftmargin=3em,
    ]
  \item
    $f(g, g'')
      \;<\; f(g', g'') 
      \;\;\;\text{for all $g'' \in [0, \widehat g]$}$
    \item
      $f(g, \AC) \;\le\; f(g', \AC)$
    \item
      $f(g, \AD) \;>\; f(g', \AD)$. 
  \end{enumerate}
\end{restatable}

\subsubsection{Analysis setup}
\label{sec:tech-overview:kigt-intro}

Given these $k$-IGT dynamics, our first main result
(Theorem~\ref{thm:kigt-main}) characterizes
the stationary and mixing properties of the distribution
over $\gtft$ strategies induced by the dynamics.
For this, let $m := \gamma \cdot n$ denote the
number of $\gtft$ nodes in the population, and fix $k \ge 2$.
We define
$\z^t := (z^t_1, \dots, z^t_k) \in \Delta^m_k$
as the \textit{count vector} specifying the number of agents with
strategy $g_i$ after the $t$'th step,
and we study the Markov chain $\{\z^t\}$.
To this end, we begin by specifying how the
transitions of the $k$-IGT dynamics map to
the transitions $\z^t \to \z^{t+1}$:
recall from Definition~\ref{def:kigt} that
following any (non-null) interaction, exactly one
$\gtft$ agent updates its parameter.
Then conditioned on an interaction at step $t$
whose initiator has strategy $g_j$
for some $j \in [k]$, 
then the coordinates of $\z^{t+1}$
can be specified by one of the following cases,
depending on the strategy of the sampled interaction partner:

\begin{multicols}{2}
\begin{enumerate}[
  label=(\alph*),
  leftmargin=2em,
  rightmargin=2em,
  itemsep=0pt]
\item
  If the second agent has strategy $\AC$ or $\gtft$,
  then for each $j \in [k]$: 
  \begin{equation*}
    z^{t+1}_i =
    \begin{aligned}
      \begin{cases}
        \;z^t_i - 1 &\text{if $i = j$ and $j < k$} \\
        \;z^t_i + 1 &\text{if $i = j+1$ and $j < k$} \\
        \;z^t_i &\text{otherwise.}
      \end{cases}
    \end{aligned}
  \end{equation*}
\item
  If the second agent has strategy $\AD$,
  then for each coordinate $i \in [k]$: 
  \begin{equation*}
    z^{t+1}_i =
    \begin{aligned}
      \begin{cases}
        \;z^t_i - 1 &\text{if $i = j$ and $j > 1$} \\
        \;z^t_i + 1 &\text{if $i = j-1$ and $j > 1$} \\
        \;z^t_i &\text{otherwise.}
      \end{cases}
    \end{aligned}
  \end{equation*}
\end{enumerate}
\end{multicols}
Given that the pair of interacting agents are sampled
uniformly at random at each time step, this implies
that the update in (a) occurs with (unconditional) probability
$(z^t_j/n) \cdot (1-\beta)$, and the update in (b)
occurs with probability $(z^t_j/n) \cdot \beta$.
Then given $\z^t$, we can summarize
all transitions $\z^t \to \z^{t+1}$
(for $\z^{t+1} \neq \z^t$) that occur with non-zero probability
as follows: for all $j \in [k-1]$,
\begin{equation}
  \begin{aligned}
    \z^{t+1}
    &\;=\;
    (z^t_1, \dots, z^t_j - 1, z^t_{j+1} + 1, \dots, z^t_k)
    \;\;\;
    \text{w.p. $\tfrac{z^t_j}{m} \cdot (1-\alpha-\beta)(1-\beta)$}
    \\
    \text{and}\;\;\;
    \z^{t+1}
    &\;=\;
    (z^t_1, \dots, z^t_j + 1, z^t_{j+1} - 1, \dots, z^t_k)
    \;\;\;
    \text{w.p. $\tfrac{z^t_{j+1}}{m} \cdot (1-\alpha-\beta)\beta$}
    \;\;.
  \end{aligned}
  \label{eq:kigt-zt}
\end{equation}
Observe that transition probabilities in~\eqref{eq:kigt-zt}
are normalized by $m$ (using the definition $n = m / (1-\alpha-\beta)$),
and that the coefficients $(1-\alpha-\beta)(1-\beta)$
and $(1-\alpha-\beta)(\beta)$ are absolute constants
with respect to the coordinates of $\z^t$. 
Thus we can view the process as
a special case of a more general class of Markov chains
$\{\x^t\}$ over $\Delta^m_k$, whose transition
probabilities (up to the absolute constant coefficients)
are of the form in expression~\eqref{eq:kigt-zt}. 
We proceed to define and analyze this more general
set of processes, from which characterizing the
stationary and mixing properties of $\{\z^t\}$ will follow. 

\subsection{High-Dimensional, Weighted Ehrenfest Processes}
\label{sec:tech-overview:ehrenfest}

We introduce and analyze a more general class
of random walks on $\Delta^m_k$,
which we refer to as 
\textit{high-dimensional, weighted Ehrenfest processes}.
Defined formally:

\begin{restatable}[$(k, a, b ,m)$-Ehrenfest Process]
  {definition}{ehrenfestkd}
  \label{def:ehrenfestkd}
  Fix $k \ge 2$, and $a, b > 0$ such that $a+b \le 1$. 
  Let $\{\x^t\}$ be the Markov chain on $\Delta^m_k$
  with transition matrix $\P: \Delta^m_k \to \Delta^m_k$,
  where for all $j \in [k-1]$ and
  $\x  = (x_1, \dots, x_k) \in \Delta^m_k$:
  \begin{align*}
    &\P\big(\x, (x_1, \dots, x_j -1, x_{j+1} + 1, \dots, x_k)\big)
    \;=\;
      p^{j, j+1}_\x
      \;:=\; 
      a \cdot \frac{x_j}{m} \\
    &\P\big(\x, (x_1, \dots, x_j + 1, x_{j+1} - 1, \dots, x_k)\big)
    \;=\;
      p^{j+1, j}_\x
      \;:=\;
      b \cdot \frac{x_{j+1}}{m} \\
    &\P(\x, \x)
    \;=\;
    p^\bot_\x
    \;:=\;
    1 - \Big(
    {\textstyle \sum_{j=1}^{k-1}}\;
    p^{j, j+1}_\x + p^{j+1, j}_\x
    \Big) \;,
  \end{align*}
  and $\P(\x, \y) = 0$ for all other $\y \in \Delta^m_k$.
  Then we call $\{\x^t\}$ the $(k, a, b, m)$-Ehrenfest process.
\end{restatable}

\paragraph{Relationship to the two-urn Ehrenfest Process}
When $k=2$ and $a=b=\frac{1}{2}$, the process
reduces to the classical \textit{Ehrenfest Urn Process}
\cite{ehrenfest1907zwei, takacs1979urn, grimmett2020probability}
from statistical physics.
Here, $m$ balls are distributed in two urns.
At each step, an urn is sampled proportionally to its load,
and with probability half, a ball from the sampled urn
is placed into the other urn.
The $(k, a, b, m)$-Ehrenfest process generalizes
this original setting to a weighted, high-dimensional regime:
we consider $m$ balls distributed over a sequence of $k$ urns,
and after sampling the $j$'th urn proportionally to its load,
a ball from urn $j$
is placed into urn $[j+1]$ with probability $a$,
and into urn $[j-1]$ with probability $b$. 
While the stationary and mixing behavior of
two-urn process (including several weighted variants)
is well-studied
\cite{kac1947random, karlin1965ehrenfest,krafft1993mean,
  dette1994generalization, 
  diaconis2011mathematics, mitzenmacher2017probability},
we give the first such analyses for the weighted,
high-dimensional analogs from Definition~\ref{def:ehrenfestkd}.

\subsubsection{Deriving the stationary distributions}
We exactly characterize the stationary distributions of
$(k, a, b, m)$-Ehrenfest processes:
we show these distributions
are \textit{multinomial} with
parameters $(p_1, \dots, p_k)$, and $m$, where
each $p_j \propto (a/b)^{j-1}$.
For $k=2$ and 3, this is obtained by
viewing the process as a weighted random walk
on a graph with vertex set $\Delta^m_k$,
and by solving the recurrences stemming from
the detailed balance equations
(these calculations are derived formally in  
Sections~\ref{app:ehrenfest:stationary-derive-k2}
and~\ref{app:ehrenfest:stationary-derive-k3}).

For higher dimensions (i.e., general $k$),
we use the form of the stationary PMFs for $k=2$ and 3
as an Ansatz for the specifying and verifying (via the
detailed balance equations) the
stationary PMF. 
Stated formally, we prove the following result,
the proof of which is given in
Appendix~\ref{app:ehrenfest:stationary-verify}.

\begin{restatable}{theorem}{ehrenfestkdstationary}
  \label{thm:ehrenfestkd:stationary}
  Fix $a, b > 0$ with $a+b \le 1$, and let $\lambda := a/b$.
  For any $k, m \ge 2$,
  let $\{\x^t\}$ be the $(k, a, b, m)$-Ehrenfest process,
  and let $\bfpi: \Delta^m_k \to [0, 1]$
  be its stationary distribution.
  Then $\bfpi$ is multinomial with parameters
  $m$ and $(p_1, \dots, p_k)$, where 
  $p_j :=
  \frac{\lambda^{j-1}}{
    \sum_{i=1}^k \lambda^{i-1}
  }
  $
  for all $j \in [k]$.
\end{restatable}
 
\subsubsection{Bounds on mixing times} 
Let $\tmix$ and $d(t)$ denote the
mixing time and distance to stationarity
(as defined in Section~\ref{sec:tech-overview:prelims})
of the $(k, a, b, m)$-Ehrenfest process.
We prove the following upper and lower bounds
on $\tmix$:

\begin{restatable}{theorem}{ehrenfestkdmixing}
  \label{thm:ehrenfestkd:mixing}
  Fix $a, b > 0$ with $a+b \le 1$,
  and $k, m \ge 2$.
  Let $\tmix$ be the mixing time of
  the $(k, a, b, m)$-Ehrenfest process.
  Then
  \begin{equation*}
    \tmix
    \;=\;
    \begin{aligned}
      \begin{cases}
        \;O\big(\min\{\frac{k}{|a-b|}, k^2 \} \cdot m \log m \big)
        \;\;\;&\text{when $a \neq b$} \\
        \;O\big( k^2 \cdot m \log m \big)
        \;\;\;&\text{when $a = b$}
      \end{cases}
    \end{aligned}
    \;\;.
  \end{equation*}
  Moreover, $\tmix = \Omega(km)$.
\end{restatable}

Here, observe that the case distinction in the 
upper bound quantifies the impact of more \textit{biased}
$a$ and $b$ parameters in speeding up
convergence to the stationary distribution, 
while the lower bound establishes a  linear dependence
on $k$ that is uniform over all $a$ and $b$. 
The full proof of the theorem is 
developed in 
Appendix~\ref{app:ehrenfest:mixing}, 
but we provide a high-level proof sketch here:

\paragraph{Proof sketch of Theorem~\ref{thm:ehrenfestkd:mixing}}
To derive an upper bound on $\tmix$, 
we introduce the following coupling:
first, let $\{X_t\}$ and $\{Y_t\}$
be random walks over $\Omega = \{1, \dots, k\}^m$.
At time $t$, we sample a coordinate $i \in [m]$ uniformly at random,
and simultaneously increment or decrement the $i$'th coordinate
of both $X_t$ and $Y_t$ (with values truncated to $[k]$)
with probability $a$ and $b$, respectively.

It is straightforward to see that
the vector of \textit{counts} of each value $j \in [k]$
in both $X_t$ and $Y_t$ evolve as
$(k, a, b, m)$-Ehrenfest processes. Then using the
standard relationship between the \textit{coupling
  time} of $\{(X_t, Y_t)\}$ (the first
$t$ when $X_t = Y_t$) and
mixing times~\cite{levin2017markov},
it suffices to probabilistically upper bound the
coupling time of the joint process. We achieve
this by estimating the time 
to coalesce each of the $m$ coordinates of the process, and this
reduces to bounding the expected absorption times
of $m$ independent (possibly) biased random walks on
$\{-k, \dots, k\}$ (which necessitates the
case distinction between $a\neq b$
and $a=b$). 
The full coupling setup is given in 
Appendix~\ref{app:ehrenfest:mixing:upper}
and requires a more careful coupling analysis 
compared to the original two-urn process.

For the $\Omega(km)$ lower bound on $\tmix$, 
we use a standard diameter lower bound 
approach~\cite{levin2017markov}, which 
says that $\tmix \ge \Omega(D)$, where
$D$ is the diameter of the graph induced
by the probability transition matrix of the
process. For this, we can derive an straightforward
estimate of $D \ge \Omega(km)$ using the structure of
the transition probabilities from Definition~\ref{def:ehrenfestkd}.
The full argument of the lower bound is given in
Appendix~\ref{app:ehrenfest:mixing:lower}.

\begin{remark}
\label{remark:ehrenfest}
  The results we establish for the $(k, a,b, m)$-Ehrenfest 
  process are general and may be of independent and broader interest. 
  Moreover, our mixing time bounds of 
  Theorem~\ref{thm:ehrenfestkd:mixing} open up several
  interesting questions.
  First, notice that our lower bounds on $\tmix$ leave open
  at least a  $O(\log m)$ gap relative to the upper bound. 
  On the other hand, $\Omega(m \log m)$ is a known lower bound
  on $\tmix$ for the original, unweighted $k=2$ process~\cite{
    diaconis1990asymptotic,levin2017markov},
  and thus we conjecture that our upper bounds on $\tmix$ for
  the general $(k, a, b, m)$ process are asymptotically optimal. 
  We leave establishing lower bounds for $\tmix$ of this same
  order as future work. 
  Additionally, the original two-urn process
  is known to exhibit a \textit{cut-off} phenomenon
  \cite{basu2014characterization},
  in which the distance to stationarity of the process
  sharply decays precisely at around $\frac{1}{2} \cdot m \log m$ 
  steps ~\cite{levin2017markov}.
  Investigating this phenomenon for the general
  $(k, a,b, m)$ process (and obtaining such exact cutoff constants
  in terms of $a$ and $b$)
  is an interesting line of future work.
\end{remark}

\subsection{Stationary Properties of $k$-IGT Dynamics}
\label{sec:tech-overview:kigt-reduce}

By combining the arguments of 
Sections~\ref{sec:tech-overview:kigt-intro}
and~\ref{sec:tech-overview:ehrenfest},
we can use the stationary and mixing time analysis
of the $k$-dimensional, weighted Ehrenfest process
to analyze the evolution of the $k$-IGT dynamics
and formally state our main results. 
Specifically, based on the transition probabilities
in~\eqref{eq:kigt-zt}, for any $k \ge 2$, 
the sequence $\{\z_t\}$ induced by the $k$-IGT dynamics is
exactly a $(k, a, b, m)$-Ehrenfest process, where
$a := \gamma (1-\beta)$,
$b := \gamma \beta$,
and $m = \gamma n$.
Thus the stationary distribution and mixing time bounds
developed in Theorems~\ref{thm:ehrenfestkd:stationary}
and~\ref{thm:ehrenfestkd:mixing}
exactly apply to the sequence $\{\z_t\}$
induced by the $k$-IGT dynamics. 

Before stating this main result,
for convenience and readability, we first summarize
in Table~\ref{table:summary-small} our core notation 
and the components of RD games in
the $(\alpha, \beta,\gamma)$ populations:

\begin{table}[h!]
  \small
\begin{center}
  \begin{tabular}{cl}
    \toprule
    Symbol & Definition\\
    \midrule
    $b, c$ & donation game reward parameters
    \\
        $\alpha$ & fraction of $\AC$ agents
    \\
    $\delta$ & continuation probability
    \\
        $\beta$ & fraction of $\AD$ agents
    \\
    $s_1$ & initial cooperation probability
    \\
        $\gamma$ & fraction of $\gtft$ agents
    \\
    $\hatg$ & maximum generosity parameter
    \\
        $k$ & number of $\gtft$ parameter values \\
\bottomrule
\end{tabular}
\end{center}
\caption{Summary of Notation in RD Games
and $(\alpha,\beta, \gamma)$ populations}
\label{table:summary-small}
\end{table}

Then the stationary and convergence properties
of the $k$-IGT dynamics are as follows:

\begin{theorem}
  \label{thm:kigt-main}
  Fix $k \ge 2$, and consider the sequence $\{\z^t\}$
  induced by the $k$-IGT dynamics
  on an $(\alpha, \beta, \gamma)$ population and an $RD$ game setting
  with maximum generosity parameter $\widehat g$.
  Then $\{\z^t\}$ converges
  to a multinomial stationary distribution $\bfpi$ with parameters
  $m$ and $(p_1, \dots, p_k)$, where each
  $
  p_j =
  \frac{(1/\beta- 1)^{(j-1)}}{\sum_{i=1}^k (1/\beta-1)^{(i-1)}}
  $
  for $j \in [k]$.
  Moreover, letting $\tmix$ denote the mixing time of $\{\z^t\}$ to $\bfpi$:
  \begin{equation*}
  \tmix
  \;\le\;
  \begin{aligned}
    \begin{cases}
      \;O\big(\min\{\frac{k}{|1-2\beta|}, k^2 \} \cdot n \log n \big)
      \;\;\;&\text{when $\beta \neq \frac{1}{2}$} \\
      \;O\big( k^2 \cdot n \log n \big)
      \;\;\;&\text{when $\beta = \frac{1}{2}$}
    \end{cases}
  \end{aligned}\;\;,
  \end{equation*}
  and $\tmix \ge \Omega(kn)$.
\end{theorem}

Observe that the mixing time of the process
speeds up in regimes where $\beta$ is bounded away
from half (i.e., the number of $\AD$ agents
is sufficiently small or sufficiently large). 
Similarly, the mean of the stationary
distribution $\E[\bfpi] = (\E[\pi_1], \dots, \E[\pi_k])$
grows increasingly \textit{less uniform} over $\calG$
in this regime of $\beta$. 
In particular, given that $\bfpi$ is a multinomial
distribution, it follows that 
$\E[\pi_j] = m \cdot p_j = \gamma n \cdot p_j$ 
for each $j \in [k]$. Thus for $\beta < 1/2$, we expect
the largest generosity parameter $g_k \in \calG$
to have the greatest adoption among $\gtft$ agents
after $\tmix$ many steps, and this mass
increases as $\beta$ grows smaller, and as the 
size $k$ of the parameter space $\calG$ increases.

\paragraph{Average stationary generosity}
Given the set of generosity values $\calG = \{g_1, \dots, g_k\}$
and any $\z = (z_1, \dots, z_k) \in \Delta^m_k$,
we define the \textit{average generosity value}
specified by $\z$ as 
$\frac{1}{m} \sum_{j\in [k]} g_j \cdot z_j$.
Then the stationary distribution
$\bfpi$ from Theorem~\ref{thm:kigt-main} allows us to
derive an \textit{average stationary generosity}
value $\wg$ for the $k$-IGT dynamics, which we define
as the average generosity value with respect to 
$\E[\bfpi]$.
We derive this value for all $k \ge 2$ in
the following proposition:

\begin{restatable}{proposition}{averagegen}
  \label{prop:avg-g}
  Fix $k \ge 2$, and let $\bfpi = (\pi_1, \dots, \pi_k)$
  denote the stationary distribution of the
  $k$-IGT dynamics from Theorem~\ref{thm:kigt-main}
  on an $(\alpha, \beta, \gamma)$ population with
  maximum generosity parameter $\hg$. 
  Let $\calG = \{g_1, \dots, g_k\}$ be the set of 
  generosity parameter values from Definition~\ref{def:kigt}.
  Let $\wg$ be the average stationary generosity of the dynamics,
  where $\widetilde g_k := \frac{1}{m} \sum_{j=1}^k g_j \cdot \E[\pi_j]$, 
  and let $\lambda := (1-\beta)/\beta$.
  Then 
  \begin{equation*}
    \wg \;=\;
    \hg \cdot 
    \bigg(
    \frac{\lambda^k}{\lambda^k - 1} -
    \Big(\frac{1}{k-1}\Big)
    \Big(\frac{\lambda}{\lambda-1}\Big)
    \Big(\frac{\lambda^{k-1}-1}{\lambda^k -1}\Big)
    \bigg)
    \;\;\text{for $\beta \neq 1/2$}\;,
  \end{equation*}
  and $\wg = \hatg / 2$ for $\beta = 1/2$. 
\end{restatable}

Roughly speaking, Proposition~\ref{prop:avg-g}
shows that for the $k$-IGT dynamics, 
$\wg \approx \hg \cdot \big(1 -\frac{\beta}{(1-2\beta)k} \big)$
when $\beta$ is bounded below $1/2$,
and $\wg \approx \hg \cdot \big(\frac{1-\beta}{(2\beta-1)k} \big)$
when $\beta$ is bounded above $1/2$.
Thus when the fraction of $\AD$ agents is sufficiently small,
the average stationary generosity approaches
the maximum generosity parameter $\widehat g$ at a rate of
$O(1/k)$, and it approaches $0$ at this same rate otherwise.
This again highlights the tradeoffs between
the size $k$ of the parameter space, and the resulting
levels of generosity induced by the dynamics. 
The proof of the proposition is given in 
Appendix~\ref{app:asg}.

\subsection{Convergence of $k$-IGT Dynamics to an
  $\eps$-Approximate Distributional Equilibrium}
\label{sec:tech-overview:convergence}

Theorem~\ref{thm:kigt-main} shows that 
the $k$-IGT converges to a stationary distribution 
$\bfpi$ and gives a bound on its mixing time. 
Our next main result shows that this stationary 
distribution corresponds to an $\eps$-approximate
distributional equilibrium. 
Specifically, we show under suitable regimes of the
RD game parameters that the (normalized) \textit{mean} of the
stationary distribution $\bfmu = \frac{1}{m} \cdot \E[\bfpi]$
of the $k$-IGT dynamics is an $\eps$-DE for $\eps = O(1/k)$.
Formally, we show the following:

\begin{restatable}{theorem}{kigteqconvergence}
  \label{thm:kigt-eq-converge}
  Let $\bfpi$ be the stationary distribution
  and $\bfmu := \frac{1}{m} \cdot \E[\bfpi]$ be the normalized
  \textit{mean} stationary distribution of the $k$-IGT dynamics
  for $\lambda := \frac{1-\beta}{\beta} \ge 2$. 
  Consider RD game settings with
  $s_1 \in [0, 1)$,
  $\frac{b}{c} > 1+ \frac{\beta c}{\gamma (1-s_1)}$
  and $\delta < \sqrt{1-\frac{\beta c}{\gamma(b-c)(1-s_1)}}$,
  and assume
  $\hatg < 1-\frac{1}{\delta}
  \Big(\frac{\beta c}{\gamma (b-c) (1-\delta) (1-s_1)} - 1\Big)$.
  Then $\bfmu$ is an $\eps$-approximate
  distributional equilibrium (Definition~\ref{def:de})
  for $\eps = O(1/k)$. 
\end{restatable}

Combined with Theorem~\ref{thm:kigt-main}, 
the result of Theorem~\ref{thm:kigt-eq-converge}
gives formal, quantitative answers to the main algorithmic questions
\textit{(Q1)} and \textit{(Q2)} stated in Section~\ref{sec:intro:de-rd}.
Specifically, for the setting where 
$\lambda = (1-\beta)/\beta \ge 2$, 
and under suitable game parameter regimes:
\begin{enumerate}[
  label=-,
  itemsep=0.5em,
  leftmargin=2em,
]
\item 
  Theorem~\ref{thm:kigt-eq-converge} gives a formal answer
  to \textit{(Q1)} by showing that the 
  mean stationary distribution $\bfmu$ of the $k$-IGT dynamics 
  is an $\eps$-approximate Distributional
  Equilibrium, for $\eps = O(1/k)$.
\item 
  Theorem~\ref{thm:kigt-main} gives a formal answer to
  \textit{(Q2)} by showing that the $k$-IGT dynamics
  converges to its stationary distribution 
  (and thus this $\eps$-DE) 
  within $O(k n \log n)$ total interactions. 
\end{enumerate}
Thus as $k$ increases, the $k$-IGT dynamics converges
to an approximate DE with a tighter approximation factor, but 
at the expense of (i) a linear increase in the 
convergence rate and (ii) a linear growth in the local memory 
required by each $\gtft$ agent (which is needed to store the set of 
$k$ generosity  parameter values $\calG = \{g_1, \dots, g_k\}$). 

Additionally, note in Theorem~\ref{thm:kigt-eq-converge}
that the constraints on the game setting
parameters require the reward
ratio $b/c$ be sufficiently large, and that the restart
probability $\delta$ and maximum generosity parameter
$\hatg$ both be bounded away from 1.
Moreover, the result is stated only for
$\lambda := (1-\beta)/\beta \ge 2$. We remark that
a similar statement to Theorem~\ref{thm:kigt-eq-converge}
can be shown for values of $\lambda$
less than 1 (e.g., $\lambda < 1/2$) that rely on
assuming an upper bound on the ratio $b/c$,
and a lower bound on $\delta$ and $\hatg$,
but for simplicity we focus on the case when
$(1-\beta)/\beta$ is greater than 1.\footnote{
  Observe that the result of 
  Theorem~\ref{thm:kigt-eq-converge}
  does not necessarily hold for $1/2 \le \lambda \le 2$. 
  For such $\lambda$ (multiplicatively) close to 1,
  the fraction $\beta$ of $\alld$ agents
  in the population is close to $1/2$. 
  Thus in this regime of $\lambda$ 
  (and depending on the game parameter settings),
  the mean of the stationary distribution of the
  corresponding Ehrenfest random walk 
  may be far from the 
  $g' \in \calG$ that maximizes
  $\E_{\S \sim \bfmu}[f(g', \S)]$, 
  and therefore the convergence to a DE may not hold. 
  In this sense, the result 
  of Theorem~\ref{thm:kigt-eq-converge} 
  requires enough ``signal'' from $\lambda$
  (i.e., needing this ratio to be bounded 
  away from 1)
  in order to hold.
}   

The full proof of Theorem~\ref{thm:kigt-eq-converge} is
developed in Appendix~\ref{app:eq}, but we provide
a high-level sketch of the proof here:

\paragraph{Proof Sketch of Theorem~\ref{thm:kigt-eq-converge}}
To prove that $\bfmu$ is an $\eps$-approximate DE,
we use the following high-level strategy:
first, recall from the Definition~\ref{def:de} that
$\bfmu$ must satisfy
\begin{equation}
  \E_{g \sim \bfmu, \;\S \sim \hatbfmu}
  \big[f(g, \S)\big]
  \;\ge\;
  \max_{g' \in \calG} \;
  \E_{\S \sim \hatbfmu}
  \big[f(g', \S)\big]
  - \eps
  \;,
  \label{eq:to-de:1}
\end{equation}
where $\hatbfmu \in \Delta(\calS)$ is
the distribution induced by $\bfmu \in \Delta(\calG)$
in the $(\alpha, \beta, \gamma)$ population.
Given that $\calG = \{g_1, \dots, g_k\}$
and $|\calG| = k$, 
the condition in expression~\eqref{eq:to-de:1}
can be equivalently written as
\begin{equation}
  \E_{g \sim \bfmu, \;\S \sim \hatbfmu}
  \big[f(g, \S)\big]
  \;\ge\;
  \max_{i \in [k]}\;
  \E_{\S \sim \hatbfmu}
  \big[f(g_i, \S)\big]
  - \eps
  \;.
  \label{eq:to-de:2}
\end{equation}
Rearranging terms, the distribution $\bfmu$ must satisfy
\begin{equation}
  \max_{i \in [k]}\;
  \E_{\S \sim \hatbfmu}
  \big[f(g_i, \S)\big]
  -
  \E_{g \sim \bfmu, \;\S \sim \hatbfmu}
  \big[f(g, \S)\big]
  \;\le\;
  \eps
  \;.
  \label{eq:to-de:3}
\end{equation}
For convenience, let $\Psi$ denote
the left hand side of expression~\eqref{eq:to-de:3},
meaning our goal in proving Theorem~\ref{thm:kigt-eq-converge}
is to show that $\Psi \le O(1/k)$. 
For this, we show 
in Proposition~\ref{prop:taylor-approx}
via a first-order Taylor approximation argument 
that, so long as the magnitude of the
second derivatives of $f(g, \S)$
(with respect to  $g$)
are uniformly bounded by a constant $L > 0$:
\begin{equation}
  \E_{g \sim \bfmu, \;\S \sim \hatbfmu} 
  \big[f(g, \S)\big]
  \;\ge\;
  \E_{\S \sim \hatbfmu} 
  \big[f(\tildeg, \S)\big]
  \;-\;
  L \cdot \Var_{g \sim \bfmu}[g] \;,
  \label{eq:to-de:4}
\end{equation}
where $\tildeg := \E_{g \sim \bfmu}[g]$
is the average stationary generosity value
from Proposition~\ref{prop:avg-g}.
Then it follows that we can further write
\begin{align}
  \Psi
  &\;:=\;
    \max_{i \in [k]}\;
    \E_{\S \sim \hatbfmu}
    \big[f(g_i, \S)\big]
    -
    \E_{g \sim \bfmu, \;\S \sim \hatbfmu}
    \big[f(g, \S)\big]
    \nonumber \\
  &\;\le\;
    \max_{i \in [k]}\;
    \E_{\S \sim \hatbfmu}
    \big[f(g_i, \S)\big]
    - 
    \E_{\S \sim \hatbfmu} 
    \big[f(\tildeg, \S)\big]
    + 
    L \cdot \Var_{g\sim \bfmu}[g]
  \nonumber \\
  &\;=\;
    \max_{i \in [k]}\;
    \E_{\S \sim \hatbfmu}
    \big[
    f(g_i, \S)
    -
    f(\tildeg, \S)
    \big]
    + 
    L \cdot \Var_{g\sim \bfmu}[g]
    \;.
    \label{eq:to-de:5}
\end{align}
Under the parameter assumptions in the statement of 
the theorem, we can then bound the two terms in 
expression~\eqref{eq:to-de:5} separately as follows:
\begin{align*}
  &\;
    \text{Propositions~\ref{prop:var-G} 
    and~\ref{prop:df2-bounded}}:
    \;\;\;
    L \cdot \Var_{g\sim \bfmu}[g] = O(1/k^2) \\
  &\;
    \text{Proposition~\ref{prop:gamma-bound}}:
    \;\;\;
    \max_{i \in [k]}\;
    \E_{\S \sim \hatbfmu}
    \big[
    f(g_i, \S)
    -
    f(\tildeg, \S)
    \big] = O(1/k) \;.
\end{align*}
The proof of the theorem then follows by combining these 
two results (i.e., into expression~\eqref{eq:to-de:5}),
and we develop the full proof in 
Appendix~\ref{app:eq}.


\section{Discussion}
\label{sec:discussion}

In this work, we initiated the study of game dynamics
and equilibrium computation in the population protocol model.
We introduced a simple family of $k$-IGT dynamics for 
a class of repeated prisoner's dilemma games, and we quantified
the convergence of these dynamics to an approximate
distributional equilibrium. In particular,
by linking the $k$-IGT dynamics to a new class of 
high-dimensional Ehrenfest processes, we obtain
convergence results that highlight the time, space, and
approximation factor tradeoffs for computing such 
global equilibria in this setting. 
Our work also opens the door for several future directions:
first, at a broad level, it would be interesting to study
game dynamics in this population setting for other
classes of games (both from classical game theory
and evolutionary games). In particular, it remains open
to quantify the tradeoffs involved in computing a
distributional equilibrium (from Definition~\ref{def:de-gen}) 
in this setting more generally. 
At a more technical level, our work also leads to 
several interesting questions on $(k, a, b, m)$-Ehrenfest
processes that are discussed in Remark~\ref{remark:ehrenfest},
and these directions are left for future work.

\bibliographystyle{alpha}
\bibliography{references.bib}

\newcommand{\etalchar}[1]{$^{#1}$}
\begin{thebibliography}{BCM{\etalchar{+}}17}

\bibitem[AAD{\etalchar{+}}06]{DBLP:journals/dc/AngluinADFP06}
Dana Angluin, James Aspnes, Zo{\"{e}} Diamadi, Michael~J. Fischer, and
  Ren{\'{e}} Peralta.
\newblock Computation in networks of passively mobile finite-state sensors.
\newblock {\em Distributed Comput.}, 18(4):235--253, 2006.

\bibitem[AAE08a]{AngluinAE08}
Dana Angluin, James Aspnes, and David Eisenstat.
\newblock Fast computation by population protocols with a leader.
\newblock {\em Distributed Computing}, 21(3):183--199, September 2008.

\bibitem[AAE08b]{AAE08}
Dana Angluin, James Aspnes, and David Eisenstat.
\newblock A simple population protocol for fast robust approximate majority.
\newblock {\em Distributed Computing}, 21(2):87--102, 2008.

\bibitem[AAE{\etalchar{+}}17]{alistarh2017time}
Dan Alistarh, James Aspnes, David Eisenstat, Rati Gelashvili, and Ronald~L
  Rivest.
\newblock Time-space trade-offs in population protocols.
\newblock In {\em Proceedings of the Twenty-Eighth Annual ACM-SIAM Symposium on
  Discrete Algorithms}, pages 2560--2579. SIAM, 2017.

\bibitem[AAG18]{AAG18}
Dan Alistarh, James Aspnes, and Rati Gelashvili.
\newblock Space-optimal majority in population protocols.
\newblock In {\em Proceedings of the 29th ACM-SIAM Symposium on Discrete
  Algorithms}, (SODA), pages 2221--2239, 2018.

\bibitem[AAL20]{amir2020approximate}
Talley Amir, James Aspnes, and John Lazarsfeld.
\newblock Approximate majority with catalytic inputs.
\newblock {\em arXiv preprint arXiv:2009.08847}, 2020.

\bibitem[AG15]{AG15}
Dan Alistarh and Rati Gelashvili.
\newblock Polylogarithmic-time leader election in population protocols.
\newblock In {\em Automata, Languages, and Programming}, pages 479--491.
  Springer, 2015.

\bibitem[AGV15]{alistarh2015fast}
Dan Alistarh, Rati Gelashvili, and Milan Vojnovi{\'c}.
\newblock Fast and exact majority in population protocols.
\newblock In {\em Proceedings of the 2015 ACM Symposium on Principles of
  Distributed Computing}, pages 47--56, 2015.

\bibitem[AH81]{axelrod1981evolution}
Robert Axelrod and William~D Hamilton.
\newblock The evolution of cooperation.
\newblock {\em Science}, 211(4489):1390--1396, 1981.

\bibitem[ALC{\etalchar{+}}17]{allen2017evolutionary}
Benjamin Allen, Gabor Lippner, Yu-Ting Chen, Babak Fotouhi, Naghmeh Momeni,
  Shing-Tung Yau, and Martin~A Nowak.
\newblock Evolutionary dynamics on any population structure.
\newblock {\em Nature}, 544(7649):227--230, 2017.

\bibitem[ATU21]{alistarh2021comparison}
Dan Alistarh, Martin T{\"o}pfer, and Przemys{\l}aw Uzna{\'n}ski.
\newblock Comparison dynamics in population protocols.
\newblock In {\em Proceedings of the 2021 ACM Symposium on Principles of
  Distributed Computing}, pages 55--65, 2021.

\bibitem[Aum87]{aumann1987correlated}
Robert~J Aumann.
\newblock Correlated equilibrium as an expression of bayesian rationality.
\newblock {\em Econometrica: Journal of the Econometric Society}, pages 1--18,
  1987.

\bibitem[BBHK21]{berenbrink2021loosely}
Petra Berenbrink, Felix Biermeier, Christopher Hahn, and Dominik Kaaser.
\newblock Loosely-stabilizing phase clocks and the adaptive majority problem.
\newblock {\em arXiv preprint arXiv:2106.13002}, 2021.

\bibitem[BCC{\etalchar{+}}13]{bournez2013population}
Olivier Bournez, J{\'e}r{\'e}mie Chalopin, Johanne Cohen, Xavier Koegler, and
  Mikael Rabie.
\newblock Population protocols that correspond to symmetric games.
\newblock {\em Int. J. Unconv. Comput.}, 9(1-2):5--36, 2013.

\bibitem[BCCK09]{bournez2009playing}
Olivier Bournez, J{\'e}r{\'e}mie Chalopin, Johanne Cohen, and Xavier Koegler.
\newblock Playing with population protocols.
\newblock {\em arXiv preprint arXiv:0906.3256}, 2009.

\bibitem[BCM{\etalchar{+}}17]{becchetti2018average}
Luca Becchetti, Andrea Clementi, Pasin Manurangsi, Emanuele Natale, Francesco
  Pasquale, Prasad Raghavendra, and Luca Trevisan.
\newblock Average whenever you meet: Opportunistic protocols for community
  detection.
\newblock {\em arXiv preprint arXiv:1703.05045}, 2017.

\bibitem[BCN{\etalchar{+}}14]{becchetti2014plurality}
Luca Becchetti, Andrea Clementi, Emanuele Natale, Francesco Pasquale, and
  Riccardo Silvestri.
\newblock Plurality consensus in the gossip model.
\newblock In {\em Proceedings of the twenty-sixth annual ACM-SIAM symposium on
  Discrete algorithms}, pages 371--390. SIAM, 2014.

\bibitem[BCN20]{becchetti2020consensus}
Luca Becchetti, Andrea Clementi, and Emanuele Natale.
\newblock Consensus dynamics: An overview.
\newblock {\em ACM SIGACT News}, 51(1):58--104, 2020.

\bibitem[BFK{\etalchar{+}}16]{Berenbrink16}
Petra Berenbrink, Tom Friedetzky, Peter Kling, Frederik Mallmann-Trenn, and
  Chris Wastell.
\newblock Plurality consensus via shuffling: Lessons learned from load
  balancing.
\newblock {\em arXiv preprint arXiv:1602.01342}, 2016.

\bibitem[BGK20]{berenbrink2020optimal}
Petra Berenbrink, George Giakkoupis, and Peter Kling.
\newblock Optimal time and space leader election in population protocols.
\newblock In {\em Proc.\ 52nd Annual ACM SIGACT Symposium on Theory of
  Computing (STOC 2020)}, pages 119--129, 2020.

\bibitem[BHP14]{basu2014characterization}
Riddhipratim Basu, Jonathan Hermon, and Yuval Peres.
\newblock Characterization of cutoff for reversible markov chains.
\newblock In {\em Proceedings of the twenty-sixth annual ACM-SIAM symposium on
  Discrete algorithms}, pages 1774--1791. SIAM, 2014.

\bibitem[BKD14]{bramoulle2014strategic}
Yann Bramoull{\'e}, Rachel Kranton, and Martin D'amours.
\newblock Strategic interaction and networks.
\newblock {\em American Economic Review}, 104(3):898--930, 2014.

\bibitem[CBL06]{cesa2006prediction}
Nicolo Cesa-Bianchi and G{\'a}bor Lugosi.
\newblock {\em Prediction, learning, and games}.
\newblock Cambridge university press, 2006.

\bibitem[CCDS17]{CCDS15}
Ho-Lin Chen, Rachel Cummings, David Doty, and David Soloveichik.
\newblock Speed faults in computation by chemical reaction networks.
\newblock {\em Distributed Computing}, 30(5):373–390, 2017.

\bibitem[CD11]{cai2011minmax}
Yang Cai and Constantinos Daskalakis.
\newblock On minmax theorems for multiplayer games.
\newblock In {\em Proceedings of the twenty-second annual ACM-SIAM symposium on
  Discrete algorithms}, pages 217--234. SIAM, 2011.

\bibitem[CDS{\etalchar{+}}13]{CDSPCSS13}
Yuan-Jyue Chen, Neil Dalchau, Niranjan Srnivas, Andrew Phillips, Luca Cardelli,
  David Soloveichik, and Georg Seelig.
\newblock Programmable chemical controllers made from dna.
\newblock {\em Nature Nanotechnology}, 8(10):755--762, 2013.

\bibitem[CDT09]{chen2009settling}
Xi~Chen, Xiaotie Deng, and Shang-Hua Teng.
\newblock Settling the complexity of computing two-player nash equilibria.
\newblock {\em Journal of the ACM (JACM)}, 56(3):1--57, 2009.

\bibitem[CKL16]{LCK16}
Luca Cardelli, Marta Kwiatkowska, and Luca Laurenti.
\newblock Programming discrete distributions with chemical reaction networks.
\newblock In {\em Proceedings of the 22nd International Conference on DNA
  Computing and Molecular Programming}, DNA22, pages 35--51. Springer, 2016.

\bibitem[DEG{\etalchar{+}}22]{doty2022time}
David Doty, Mahsa Eftekhari, Leszek G{\k{a}}sieniec, Eric Severson, Przemyslaw
  Uzna{\'n}ski, and Grzegorz Stachowiak.
\newblock A time and space optimal stable population protocol solving exact
  majority.
\newblock In {\em 2021 IEEE 62nd Annual Symposium on Foundations of Computer
  Science (FOCS)}, pages 1044--1055. IEEE, 2022.

\bibitem[Det94]{dette1994generalization}
Holger Dette.
\newblock On a generalization of the ehrenfest urn model.
\newblock {\em Journal of applied probability}, 31(4):930--939, 1994.

\bibitem[DGG{\etalchar{+}}02]{dyer2002convergence}
Martin Dyer, Leslie~Ann Goldberg, Catherine Greenhill, Gabriel Istrate, and
  Mark Jerrum.
\newblock Convergence of the iterated prisoner's dilemma game.
\newblock {\em Combinatorics, Probability and Computing}, 11(2):135--147, 2002.

\bibitem[DGM90]{diaconis1990asymptotic}
Persi Diaconis, Ronald~L. Graham, and John~A. Morrison.
\newblock Asymptotic analysis of a random walk on a hypercube with many
  dimensions.
\newblock {\em Random Structures \& Algorithms}, 1(1):51--72, 1990.

\bibitem[DGP09]{daskalakis2009complexity}
Constantinos Daskalakis, Paul~W Goldberg, and Christos~H Papadimitriou.
\newblock The complexity of computing a nash equilibrium.
\newblock {\em Communications of the ACM}, 52(2):89--97, 2009.

\bibitem[Dia11]{diaconis2011mathematics}
Persi Diaconis.
\newblock The mathematics of mixing things up.
\newblock {\em Journal of Statistical Physics}, 144(3):445--458, 2011.

\bibitem[DK18]{DBLP:conf/stoc/DudekK18}
Bartłomiej Dudek and Adrian Kosowski.
\newblock Universal protocols for information dissemination using emergent
  signals.
\newblock In {\em STOC}, pages 87--99, 2018.

\bibitem[Dot14]{Doty14}
David Doty.
\newblock Timing in chemical reaction networks.
\newblock In {\em Proceedings of the Twenty-Fifth Annual ACM-SIAM Symposium on
  Discrete Algorithms}, SODA '14, pages 772--784. SIAM, 2014.

\bibitem[DP15]{daskalakis2015approximate}
Constantinos Daskalakis and Christos~H Papadimitriou.
\newblock Approximate nash equilibria in anonymous games.
\newblock {\em Journal of Economic Theory}, 156:207--245, 2015.

\bibitem[DS18]{DS18}
David Doty and David Soloveichik.
\newblock Stable leader election in population protocols requires linear time.
\newblock {\em Distributed Computing}, 31(4):257--271, 2018.

\bibitem[DV12]{DV12}
Moez Draief and Milan Vojnovic.
\newblock Convergence speed of binary interval consensus.
\newblock {\em SIAM Journal on Control and Optimization}, 50(3):1087--1109,
  2012.

\bibitem[EEA07]{ehrenfest1907zwei}
Paul Ehrenfest and Tatjana Ehrenfest-Afanassjewa.
\newblock {\em {\"U}ber zwei bekannte Einw{\"a}nde gegen das Boltzmannsche
  H-Theorem}.
\newblock Hirzel, 1907.

\bibitem[EL08]{etessami2008computational}
Kousha Etessami and Andreas Lochbihler.
\newblock The computational complexity of evolutionarily stable strategies.
\newblock {\em International Journal of Game Theory}, 37:93--113, 2008.

\bibitem[FJ06]{fischer2006self}
Michael Fischer and Hong Jiang.
\newblock Self-stabilizing leader election in networks of finite-state
  anonymous agents.
\newblock In {\em Principles of Distributed Systems}, pages 395--409. Springer,
  2006.

\bibitem[GNW16]{giakkoupis2016asynchrony}
George Giakkoupis, Yasamin Nazari, and Philipp Woelfel.
\newblock How asynchrony affects rumor spreading time.
\newblock In {\em Proceedings of the 2016 ACM Symposium on Principles of
  Distributed Computing}, pages 185--194, 2016.

\bibitem[GP16]{GP16}
Mohsen Ghaffari and Merav Parter.
\newblock A polylogarithmic gossip algorithm for plurality consensus.
\newblock In {\em Proceedings of the 35th ACM Symposium on Principles of
  Distributed Computing}, (PODC), pages 117--126, 2016.

\bibitem[GPD20]{golowich2020tight}
Noah Golowich, Sarath Pattathil, and Constantinos Daskalakis.
\newblock Tight last-iterate convergence rates for no-regret learning in
  multi-player games.
\newblock {\em Advances in neural information processing systems},
  33:20766--20778, 2020.

\bibitem[GS18]{GS18}
Leszek G{\k{a}}sieniec and Grzegorz Stachowiak.
\newblock Fast space optimal leader election in population protocols.
\newblock In {\em Proceedings of the 29th ACM-SIAM Symposium on Discrete
  Algorithms}, (SODA), pages 2653--2667, 2018.

\bibitem[GS20]{grimmett2020probability}
Geoffrey Grimmett and David Stirzaker.
\newblock {\em Probability and random processes}.
\newblock Oxford university press, 2020.

\bibitem[HNS13]{hilbe2013evolution}
Christian Hilbe, Martin~A Nowak, and Karl Sigmund.
\newblock Evolution of extortion in iterated prisoner’s dilemma games.
\newblock {\em Proceedings of the National Academy of Sciences},
  110(17):6913--6918, 2013.

\bibitem[HS98]{hofbauer1998evolutionary}
Josef Hofbauer and Karl Sigmund.
\newblock {\em Evolutionary games and population dynamics}.
\newblock Cambridge university press, 1998.

\bibitem[HSP23]{hu2023best}
Shuyue Hu, Harold Soh, and Georgios Piliouras.
\newblock The best of both worlds in network population games: Reaching
  consensus and convergence to equilibrium.
\newblock In {\em Thirty-seventh Conference on Neural Information Processing
  Systems}, 2023.

\bibitem[IJCN15]{ibsen2015computational}
Rasmus Ibsen-Jensen, Krishnendu Chatterjee, and Martin~A Nowak.
\newblock Computational complexity of ecological and evolutionary spatial
  dynamics.
\newblock {\em Proceedings of the National Academy of Sciences},
  112(51):15636--15641, 2015.

\bibitem[Kac47]{kac1947random}
Mark Kac.
\newblock Random walk and the theory of brownian motion.
\newblock {\em The American Mathematical Monthly}, 54(7P1):369--391, 1947.

\bibitem[KM65]{karlin1965ehrenfest}
Samuel Karlin and James McGregor.
\newblock Ehrenfest urn models.
\newblock {\em Journal of Applied Probability}, 2(2):352--376, 1965.

\bibitem[KS93]{krafft1993mean}
Olaf Krafft and Martin Schaefer.
\newblock Mean passage times for tridiagonal transition matrices and a
  two-parameter ehrenfest urn model.
\newblock {\em Journal of Applied Probability}, 30(4):964--970, 1993.

\bibitem[LHN05]{lieberman2005evolutionary}
Erez Lieberman, Christoph Hauert, and Martin~A Nowak.
\newblock Evolutionary dynamics on graphs.
\newblock {\em Nature}, 433(7023):312--316, 2005.

\bibitem[LP17]{levin2017markov}
David~A Levin and Yuval Peres.
\newblock {\em Markov chains and mixing times}, volume 107.
\newblock American Mathematical Soc., 2017.

\bibitem[LZMJ20]{lin2020finite}
Tianyi Lin, Zhengyuan Zhou, Panayotis Mertikopoulos, and Michael Jordan.
\newblock Finite-time last-iterate convergence for multi-agent learning in
  games.
\newblock In {\em International Conference on Machine Learning}, pages
  6161--6171. PMLR, 2020.

\bibitem[Mar09]{marshall2009donation}
James~AR Marshall.
\newblock The donation game with roles played between relatives.
\newblock {\em Journal of Theoretical Biology}, 260(3):386--391, 2009.

\bibitem[Mol85]{molander1985optimal}
Per Molander.
\newblock The optimal level of generosity in a selfish, uncertain environment.
\newblock {\em Journal of Conflict Resolution}, 29(4):611--618, 1985.

\bibitem[MU17]{mitzenmacher2017probability}
Michael Mitzenmacher and Eli Upfal.
\newblock {\em Probability and computing: Randomization and probabilistic
  techniques in algorithms and data analysis}.
\newblock Cambridge university press, 2017.

\bibitem[MVCS12]{martinez2012generosity}
Luis~A Martinez-Vaquero, Jose~A Cuesta, and Angel Sanchez.
\newblock Generosity pays in the presence of direct reciprocity: A
  comprehensive study of 2$\times$ 2 repeated games.
\newblock {\em PLoS One}, 7(4):e35135, 2012.

\bibitem[Nas51]{nash1951non}
John Nash.
\newblock Non-cooperative games.
\newblock {\em Annals of mathematics}, pages 286--295, 1951.

\bibitem[Now06]{nowak2006evolutionary}
Martin~A Nowak.
\newblock {\em Evolutionary dynamics: exploring the equations of life}.
\newblock Harvard university press, 2006.

\bibitem[NRTV07]{nisan2007algorithmic}
Noam Nisan, Tim Roughgarden, Eva Tardos, and Vijay~V Vazirani.
\newblock Algorithmic game theory, 2007.
\newblock {\em Book available for free online}, 2007.

\bibitem[NS92]{nowak1992tit}
Martin~A Nowak and Karl Sigmund.
\newblock Tit for tat in heterogeneous populations.
\newblock {\em Nature}, 355(6357):250--253, 1992.

\bibitem[NS93]{nowak1993strategy}
Martin Nowak and Karl Sigmund.
\newblock A strategy of win-stay, lose-shift that outperforms tit-for-tat in
  the prisoner's dilemma game.
\newblock {\em Nature}, 364(6432):56--58, 1993.

\bibitem[OR94]{osborne1994course}
Martin~J Osborne and Ariel Rubinstein.
\newblock {\em A course in game theory}.
\newblock MIT press, 1994.

\bibitem[Owe13]{owen2013game}
Guillermo Owen.
\newblock {\em Game theory}.
\newblock Emerald Group Publishing, 2013.

\bibitem[PVV09]{PVV09}
Etienne Perron, Dinkar Vasudevan, and Milan Vojnovic.
\newblock Using three states for binary consensus on complete graphs.
\newblock In {\em INFOCOM 2009, IEEE}, pages 2527--2535. IEEE, 2009.

\bibitem[Rob84]{robert1984evolution}
Axelrod Robert.
\newblock The evolution of cooperation.
\newblock {\em Basic Books}, 1984.

\bibitem[Rou16]{roughgarden2016twenty}
Tim Roughgarden.
\newblock {\em Twenty lectures on algorithmic game theory}.
\newblock Cambridge University Press, 2016.

\bibitem[San10]{sandholm2010population}
William~H Sandholm.
\newblock {\em Population games and evolutionary dynamics}.
\newblock MIT press, 2010.

\bibitem[Smi82]{smith1982evolution}
John~Maynard Smith.
\newblock {\em Evolution and the Theory of Games}.
\newblock Cambridge University Press, 1982.

\bibitem[SP13]{stewart2013extortion}
Alexander~J Stewart and Joshua~B Plotkin.
\newblock From extortion to generosity, evolution in the iterated prisoner’s
  dilemma.
\newblock {\em Proceedings of the National Academy of Sciences},
  110(38):15348--15353, 2013.

\bibitem[Tak79]{takacs1979urn}
Lajos Tak{\'a}cs.
\newblock On an urn problem of paul and tatiana ehrenfest.
\newblock In {\em Mathematical Proceedings of the Cambridge Philosophical
  Society}, volume~86, pages 127--130. Cambridge University Press, 1979.

\end{thebibliography}

\newpage
\appendix
\section{Details on high-dimensional, weighted Ehrenfest processes}
\label{app:ehrenfest}

In this section, we provide more
details on the $(k, a, b, m)$-Ehrenfest processes
introduced in Section~\ref{sec:tech-overview:ehrenfest}.
In particular, in
Sections~\ref{app:ehrenfest:stationary-derive-k2}
and ~\ref{app:ehrenfest:stationary-derive-k3}
we derive the stationary distributions
for the process when $k=2$ and $k=3$, respectively.
We use the form of the PMFs of these distributions
to derive expressions for the PMFs
for larger (general) $k$, and
we verify their stationarity in the proof of
Theorem~\ref{thm:ehrenfestkd:stationary} in
Section~\ref{app:ehrenfest:stationary-verify}. 
In Section~\ref{app:ehrenfest:mixing}, we
develop the proof of Theorem~\ref{thm:ehrenfestkd:mixing},
which gives a bound on the mixing time of the process,
which is based on the coupling introduced in
Section~\ref{sec:tech-overview:ehrenfest}. 

For convenience, we restate the definition of the process here:

\ehrenfestkd*

\subsection{Deriving the stationary distribution for $k=2$}
\label{app:ehrenfest:stationary-derive-k2}

In this section we derive the PMF of the
stationary distribution of the $(2, a, b, m)$-Ehrenfest process,
for which we use the following, more convenient notation:

\paragraph{Setup and notation}
Let $\{\x^t\}$ denote the process, and recall
each $\x^t \in \Delta^m_k$. Let $\bfpi: \Delta^m_2 \to [0, 1]$
denote its stationary distribution.
Based on the transitions from Definition~\ref{def:ehrenfestkd},
it is easy to see that $\{\x^t\}$ is irreducible
and aperiodic, and thus $\bfpi$ is the unique
stationary distribution of the process.

In order to specify the PMF of $\bfpi$, we
restrict our view of the process $\{x^t\}$
to its first coordinate $i \in \calZ := \{0, 1, \dots, m\}$.
In particular, we define the function
$\pi: \calZ \to [0, 1]$ such that
$\pi(i) = \bfpi((i, m-i))$ for all $i \in \calZ$.
In other words, we project the set $\Delta^m_2$
onto the line $\calZ = \{0, \dots, m\}$. 
We can then define the transition matix
$P: \calZ \times \calZ \to [0, 1]$ over pairs of
points in $\calZ$ that is induced by
the matrix $\P$ for the process $\{\x^t\}$.
Thus for any $x, y \in \calZ$, the
entries of $P$ are given by:
\begin{equation}
  P\big(x, y\big)
  \;=\;
  \begin{aligned}
    \begin{cases}
      \;\text{if }
      y = [x+1]:
      &b \tfrac{m-x}{m} \\
      \;\text{if }
      y = [x-1]:
      &a \tfrac{x}{m} \\
      \;\text{if }
      y = x:
      &1 - b + (b-a)\tfrac{x}{m} \\
      \;\text{otherwise}:
      &0
    \end{cases}
  \end{aligned}
  \label{eq:P-k2-def}
  \;\;.
\end{equation}
Thus matrix $P$ specifies the
transition probabilities over the number line $\calZ$
for the \textit{first coordinate} of the process $\{\x^t\}$
(and this entirely specifies the process
given that each $\x^t \in \Delta^m_2$).

Then in the following proposition,
we derive the PMF of $\pi$ by solving the
system of equations arising from the detailed
balance equations. The PMF of $\bfpi$ is
then recovered immediately based on
the one-to-one relationship between $\pi$ and $\bfpi$.

\begin{proposition}
  \label{prop:ehrenfest2d:stationary}
  Fix $a, b > 0$ with $a + b \le 1$.
  Let $P: \calZ \times \calZ \to [0, 1]$
  denote the transition matrix induced
  by the $(2, a, b, m)$-Ehrenfest process
  (as defined in expression~\eqref{eq:P-k2-def})
  and let $\pi: \calZ \to [0, 1]$
  be the stationary distribution of $P$.
  Letting $\lambda := a/b$, we have for all $i \in \calZ$:
    \begin{equation*}
    \pi(i)
    \;=\;
    \frac{\lambda^{m-i}}{(1+\lambda)^m} \cdot \binom{m}{i} \;.
  \end{equation*}
\end{proposition}

\begin{proof}
  We solve for the $\pi$ based on the
  system of equations stemming from the detailed
  balance equations
  \begin{equation*}
    \pi(x) \cdot P(x, y)
    \;=\;
    \pi(y) \cdot P(y, x)
    \;,
  \end{equation*}
  which must hold for all $x, y \in \calZ$.
  Then for any $i\in \calZ$, we can recursively
  substitute the equations along the path
  $i, i-1, \dots, 1,  0$ to find
  \begin{equation*}
    \pi(i)
    \;=\;
    \pi(0) \cdot 
    \frac{
      P(0, 1) \cdot P(1, 2) \cdots P(i-1, i)
    }{
      P(i, i-1)\cdot P(i-1, i-2) \cdots P(1, 0)
    }
    \;.
  \end{equation*}
  It then follows from the entries of $P$
  defined in expression~\eqref{eq:P-k2-def} that
  we can write
  \begin{equation}
    \pi(i)
    \;=\;
    \pi(0)
    \cdot
    \frac{
      bm \cdot b(m-1) \cdots b(m-i+1)
    }{
      ai \cdot a(i-1) \cdots a
    }
    \;=\;
    \pi(0)
    \cdot
    \Big(\frac{b}{a}\Big)^i \cdot \frac{m!}{(m-i)! \cdot i!}
    \;.
    \label{eq:k2-pi-i}
  \end{equation}
  Given that $\sum_{i \in \calZ} \pi(i) = 1$
  (as $\pi$ is a distribution),
  and recalling that $\lambda = a/b$, we then have
  \begin{equation*}
    1 \;=\;
    \sum_{i=0}^m
    \pi(0)
    \cdot
    \Big(\frac{1}{\lambda}\Big)^i
    \cdot
    \frac{m!}{(m-i)! \cdot i!}
    \;=\;
    \pi(0)
    \cdot
    \Big(
    1 + \frac{1}{\lambda}
    \Big)^m
    \;=\;
    \pi(0)
    \cdot
    \Big(\frac{1+\lambda}{\lambda}\Big)^m
    \;,
  \end{equation*}
  where the second equality is due to the binomial theorem.
  It follows that
  $\pi (0) =  \lambda^m / (1 + \lambda)^m$,
  and substituting this into
  expression~\eqref{eq:k2-pi-i} for each
  $i \in \calZ$ yields the statement of the proposition.
\end{proof}

\begin{remark}
  Given the relationship between $\pi$ and $\bfpi$,
  it follows that for any $\x = (x_1, x_2) \in \Delta^m_k$:
  \begin{equation*}
    \bfpi(\x)
    \;=\;
    \frac{\lambda^{x_2}}{(1+\lambda)^m}
    \cdot \binom{m}{x_1}
    \;.
  \end{equation*}
  In other words, the stationary distribution
  of the $(2, a, b, m)$-Ehrenfest process is
  a \textit{binomial} distribution with parameters
  $m$ and $p = 1/(1 + \lambda)$.
\end{remark}

\subsection{Deriving the stationary distribution for $k=3$}
\label{app:ehrenfest:stationary-derive-k3}

We now derive the stationary distribution
for the $(3, a, b, m)$-Ehrenfest process.
Our approach is similar in spirit to
the $k=2$ case, but we require a more careful
technique in order to solve for the stationary PMF
in this higher-dimensional regime.
For this, we begin by introducing a more
convenient notation (similar to the notation
used for the $k=2$ case). 

\paragraph{Setup and notation}
Let $\{\x^t\}$ denote the $(3, a, b, m)$-Ehrenfest process,
and recall that its stationary distribution $\bfpi$ is defined
over the space $\Delta^m_3$.
It is again easy to see that the Markov chain $\{\x^t\}$ is irreducible,
and thus $\bfpi$ is its unique stationary distribution. 
Similar to the $k=2$, when deriving the PMF for $\bfpi$,
it is more convenient to work over the (equivalent) set
of points that are defined explicitly in a $(k-1)$-dimensional space. 
\begin{equation*}
  \calX
  \;:=\;
  \Big\{
  \tp{i}{j} \;:\;  0 \le i, j \le m \text{ and } i+j \le m
  \Big\} \;,
\end{equation*}
where we use the one-to-one mapping
$\tp{i}{j}\mapsto (i, j, m-i-j) \in \Delta^m_3$
for all $\tp{i}{j} \in \calX$.
Roughly speaking, we can view $\calX$
as simply the set $\Delta^m_3$ embedded onto
the two-dimensional plane. 

Then we define $\pi : \calX \to [0, 1]$
as the PMF over $\calX$ such that
for all $(i, j, k) \in \Delta^m_3$,
we have $\pi\tp{i}{j} = \bfpi((i, j, k))$. 
Moreover, we let $P: \calX \times \calX \to [0, 1]$ denote
the probability transition function over pairs of
points in $\calX$ that is induced by the transition matrix
$\P: \Delta^m_3 \times \Delta^m_3 \to [0, 1]$ of
the process $\{\x^t\}$. Then for any $(x_1, x_2), (y_1, y_2) \in \calX$,
the entries of $P$ can be summarized as
\begin{equation}
  P\Big(
  \tp{x_1}{x_2},
  \tp{y_1}{y_2}
  \Big)
  \;=\;
  \begin{aligned}
    \begin{cases}
      \;\text{if }
      \tp{y_1}{y_2} = \tp{[x_1-1]}{[x_2 + 1]}: 
      &a\tfrac{x_1}{m} 
      \vspace*{0.2em} \\
      \;\text{if }
      \tp{y_1}{y_2} = \tp{[x_1 + 1]}{[x_2 - 1]}:
      &b\tfrac{x_2}{m}
      \vspace*{0.2em} \\
      \;\text{if }
      \tp{y_1}{y_2} = \tp{x_1}{[x_2 - 1]}:
      &a\tfrac{x_2}{m}
      \vspace*{0.2em} \\
      \;\text{if }
      \tp{y_1}{y_2} = \tp{x_1}{[x_2 + 1]}:
      &b\big(1- \tfrac{x_1 + x_2}{m}\big)
      \vspace*{0.2em} \\
      \;\text{if }
      \tp{y_1}{y_2} = \tp{x_1}{x_2}:
      &1- b\big(1- \tfrac{x_1}{m}\big)
      - a \tfrac{x_1 + x_2}{m}
      \vspace*{0.2em} \\
      \;\text{otherwise}:
      &0 
    \end{cases}
  \end{aligned}
  \;\;\;\;.
  \label{eq:P-k3-def}
\end{equation}
For concreteness, Figure~\ref{fig:k3-ehrenfest}
shows an example of the space $\calX$
and the transitions with non-zero probability under $P$
when $k=3$ and $m=3$. 
We then prove the following proposition
specifying the PMF of $\pi$:

\begin{figure}[htb!]
  \centering
  \includegraphics[width=0.3\textwidth]{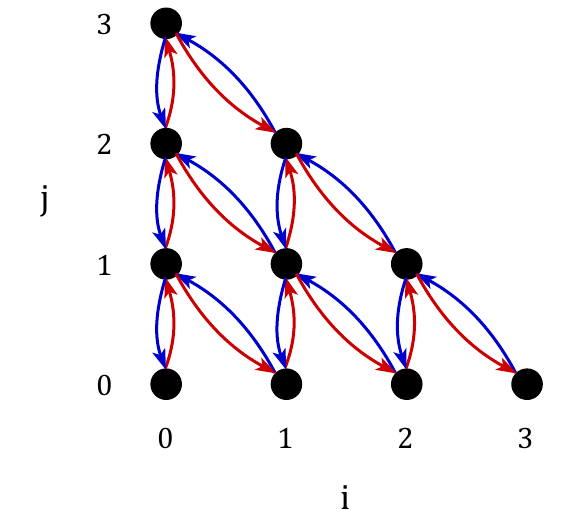}
  \caption{\small
    Example when $k=3$ and $m=3$ of
    the space $\calX$ and the set of
    non-zero transitions specified by $P$ 
    (denoted by a directed edge between two nodes).
    Transitions colored blue have coefficient $a$,
    and transitions colored red have coefficient $b$.
    Considering the outward transition edges 
    from a node $\tp{i}{j} \in \calX$, 
    downward vertical transitions are weighted by $j/n$,
    upward vertical transitions are weighted by $(1-(i+j)/n)$,
    downward diagonal transitions are weighted by
    $j/n$, and upward diagonal transitions 
    are weighted by $i/n$. 
    }
  \label{fig:k3-ehrenfest}
\end{figure}

\begin{proposition}
  \label{prop:ehrenfest3d:stationary}
  Fix $a, b > 0$ with $a + b \le 1$,
  let $P: \calX \times \calX \to [0, 1]$ denote
  the transition matrix induced by the
  $(3, a, b, m)$-Ehrenfest process (as defined
  in expression~\eqref{eq:P-k3-def}),
  and let $\pi: \calX \to [0, 1]$ be stationary distribution
  of $P$. Then letting $\lambda := a/b$, we have
  \begin{equation*}
    \pi\tp{i}{j}
    \;=\;
    \frac{\lambda^{2(m-i-j)}\lambda^j}{(1+ \lambda + \lambda^2)^m}
    \cdot
    \binom{m}{i, j, (m-i-j)}
  \end{equation*}
  for all $x = \tp{i}{j} \in\calX$. 
\end{proposition}

\begin{proof}
  It suffices to solve for $\pi$ from the system of equations
  arising from the detailed balance equations:
  \begin{equation*}
    \pi\tp{x_1}{y_1} \cdot 
    \Pp{\tp{x_1}{y_1}}{\tp{x_2}{y_2}}
    \;=\;
    \pi\tp{x_2}{y_2} \cdot 
    \Pp{\tp{x_2}{y_2}}{\tp{x_1}{y_1}} \;,
  \end{equation*}
  which must hold for all pairs $(x_1,y_1), (x_2, y_2) \in \calX$. 
  Using the recursive structure of the transition
  probabilities $P$ defined in expression~\eqref{eq:P-k3-def},
  we can express each $\pi\tp{i}{j}$ in terms
  of $\pi\tp{0}{m}$ for all $\tp{i}{j} \in \calX$
  using the following formulations:
  \begin{enumerate}[
    label={(\alph*)},
    leftmargin=2em,
    rightmargin=2em,
    itemsep=0pt,
    ]
  \item
    For any $j \in \{0, \dots, m\}$:
    \begin{equation*}
      \pi\tp{0}{j}
      \;=\;
      \pi\tp{0}{m} \cdot
      \frac{
        \Pp{\tp{0}{m}}{\tp{0}{m-1}}
        \cdot
        \Pp{\tp{0}{m-1}}{\tp{0}{m-2}}
        \cdots
        \Pp{\tp{0}{j+1}}{\tp{0}{j}}
      }{
        \Pp{\tp{0}{j}}{\tp{0}{j+1}}
        \cdot
        \Pp{\tp{0}{j+1}}{\tp{0}{j+2}}
        \cdots
        \Pp{\tp{0}{m-1}}{\tp{0}{m}}
      }
      \;\;.
    \end{equation*}
    \item
      For any $i \in \{1, \dots, m\}$ and $j \in \{0, \dots, m\}$:
      \begin{equation*}
        \pi\tp{i}{j}
        \;=\;
        \pi\tp{0}{j+i} \cdot
        \frac{
          \Pp{\tp{0}{j+i}}{\tp{1}{j+i-1}}
          \cdot
          \Pp{\tp{1}{j+i-1}}{\tp{2}{j+i-2}}
          \cdots
          \Pp{\tp{i-1}{j+1}}{\tp{i}{j}}
        }{
          \Pp{\tp{i}{j}}{\tp{i-1}{j+1}}
          \cdot
          \Pp{\tp{i-1}{j+1}}{\tp{i-2}{j+2}}
          \cdots
          \Pp{\tp{1}{j+1-i}}{\tp{0}{j+i}}
        }
        \;\;,
      \end{equation*}
      from which it follows that $\pi\tp{i}{j}$
      can be expressed in terms of $\pi\tp{0}{m}$
      using the expression for $\pi\tp{0}{i+j}$ from (a).
    \end{enumerate}
    Note that by viewing the transitions specified by $P$
    as directed edges over the vertex set $\calX$
    (i.e., as in the example of Figure~\ref{fig:k3-ehrenfest})
    the formulation in (a) follows by
    recursively substituting the expressions for
    $\pi\tp{0}{j}$ from the detailed balance equation along
    the path from $\tp{0}{j}$ to $\tp{0}{m}$ using
    only \textit{vertical} edges. 
    Similarly, the formulation in (b) follows by considering
    such paths from $\pi\tp{i}{j}$ to $\pi\tp{0}{i+j}$
    using only \textit{diagonal} edges.
    Then using the entries of $P$ defined in \eqref{eq:P-k3-def},
    we find for part (a) that, for any $j \in \{0, \dots, m\}$:
    \begin{equation}
      \pi\tp{0}{j}
      \;=\;
      \pi\tp{0}{m} \cdot
      \frac{
        am 
        \cdot
        a(m-1)
        \cdots
        a(j+1)
      }{
        b(m-j)
        \cdot
        b(m-(j+1))
        \cdots
        b
      }
      \;=\;
      \pi\tp{0}{m} \cdot
      \Big(\frac{a}{b}\Big)^{m-j}
      \cdot
      \frac{m!}{j! \cdot (m-j)!}
      \;\;.
      \label{eq:k3:vertical}
    \end{equation}
    Similarly, for part (b), we find for any $i \in \{1, \dots, m\}$
    and $j \in \{0, \dots, m\}$:
    \begin{equation}
      \pi\tp{i}{j}
      \;=\;
      \pi\tp{0}{j+i} \cdot
      \frac{
        b(j+i)
        \cdot
        b(j+i-1)
        \cdots
        b(j+1)
      }{
        ai
        \cdot
        a(i-1)
        \cdots
        a
      }
      \;=\;
      \pi\tp{0}{j+i} \cdot
      \Big(\frac{b}{a}\Big)^{i}
      \cdot 
      \frac{(j+i)!}{j! \cdot i!}
      \;\;.
      \label{eq:k3:diagonal:pre}
    \end{equation}
    Then substituting the expression for $\pi\tp{0}{j+i}$
    from~\eqref{eq:k3:vertical}, we can further simplify and write
    \begin{align}
      \pi\tp{i}{j}
      &\;=\;
        \pi\tp{0}{m} \cdot
        \Big(\frac{a}{b}\Big)^{m-(i+j)}
        \cdot
        \frac{m!}{(j+i)! \cdot (m-(i+j))!}
        \cdot
        \Big(\frac{b}{a}\Big)^{i}
        \cdot 
        \frac{(j+i)!}{j! \cdot i!}
        \nonumber \\
      &\;=\;
        \pi\tp{0}{m} \cdot
        \Big(\frac{a}{b}\Big)^{m-2i-j}
        \cdot
        \frac{m!}{(m-i-j)! \cdot i! \cdot j!}
        \;.
      \label{eq:k3:diagonal}
    \end{align}   
  Now using the fact that $\pi$ is a distribution
  and thus $\sum_{x \in \calX} \pi(x) = 1$,
  and recalling that $\lambda = (a/b)$, we can
  use the expressions~\eqref{eq:k3:vertical}
  and~\eqref{eq:k3:diagonal} to write:
  \begin{align}
    1
    &\;=\;
      \pi\tp{0}{m}
      + 
      \sum_{j=0}^{m-1}
      \bigg(
      \pi\tp{0}{j}
      +
      \sum_{i=1:\; i+j \le m}^m
      \pi\tp{i}{j}
      \bigg)
      \nonumber \\
    &\;=\;
      \pi\tp{0}{m}
      + 
      \sum_{j=0}^{m-1}
      \bigg(
      \pi\tp{0}{m}
      \cdot
      \lambda^{m-j}
      \cdot
      \tfrac{m!}{j!\cdot (m-j)!}
      +
      \sum_{i=1:\; i+j \le m}^m
      \;
      \pi\tp{0}{m}
      \cdot
      \lambda^{m-2i-j}
      \cdot
      \tfrac{m!}{(m-i-j)!\cdot i! \cdot j!}
      \bigg)
      \nonumber \\
    &\;=\;
      \sum_{(i,j) \in \calX}
      \pi\tp{0}{m}
      \cdot
      \lambda^{m-2i-j}
      \cdot
      \tfrac{m!}{(m-i-j)! \cdot i! \cdot j!}
      \;.
      \label{eq:k3:solve:1}
  \end{align}
  Now observe by the multinomial theorem that we can write
  \begin{equation*}
    \sum_{(i,j) \in \calX}
    \;
    \lambda^{m-2i-j}
    \cdot
    \frac{m!}{(m-i-j)! \cdot i! \cdot j!}
    \;=\;
    \bigg(
    \frac{1 + \lambda + \lambda^2}{\lambda}
    \bigg)^m
    \;.
  \end{equation*}
  Then it follows from expression~\eqref{eq:k3:solve:1} that
  \begin{equation*}
    \pi\tp{0}{m}
    \;=\;
    \frac{\lambda^m}{(1+ \lambda + \lambda^2)^m} \;.
  \end{equation*}
  Finally, substituting this expression for $\pi\tp{0}{m}$
  back into equations~\eqref{eq:k3:vertical}
  and~\eqref{eq:k3:diagonal} then
  specifies the mass $\pi\tp{i}{j}$ for general $\tp{i}{j} \in \calX$,
  which concludes the proof. 
\end{proof}

\begin{remark}
  Using the relationship between $\pi$ and $\bfpi$ mentioned
  earlier, it follows immediately from
  Proposition~\ref{prop:ehrenfest3d:stationary}
  that for the $(3, a, b, m)$-Ehrenfest process, 
  the PMF of $\bfpi$ is specified by
  \begin{equation*}
    \bfpi(\x)
    \;=\;
    \frac{\lambda^{2x_3 + x_2}}{(1 + \lambda + \lambda^2)^m}
    \cdot
    \binom{m}{x_1, x_2, x_3}
  \end{equation*}
  for all $\x = (x_1, x_2, x_3) \in \Delta^m_k$,
  where $\lambda = (a/b)$.
  In other words, $\bfpi$ is a
  multinomial distribution with parameters $m$ and
  $(p_1, p_2, p_3)$, where 
  \begin{equation*}
    p_1 := \frac{1}{1+ \lambda + \lambda^2}\;,
    \;\;\;\;
    p_2 := \frac{\lambda}{1+ \lambda + \lambda^2}\;,
    \;\;\;\;\text{and}\;\;\;
    p_3 := \frac{\lambda^2}{1+ \lambda + \lambda^2}
    \;.
  \end{equation*}
\end{remark}

\subsection{Verifying the stationary distribution for general $k$}
\label{app:ehrenfest:stationary-verify}

We use the stationary PMFs for the $(2, a, b, m)$ and
$(3, a, b, m)$-Ehrenfest processes found in
Sections~\ref{app:ehrenfest:stationary-derive-k2}
and~\ref{app:ehrenfest:stationary-derive-k3} as
an Ansatz for specifying the PMF of the stationary
distribution for general $k$. We can then prove
stationarity by verifying this PMF satisfies the
the detailed balance equations.
In particular, we show that the PMF of $\bfpi$
for the $(k, a, b, m)$-Ehrenfest process is of the form:
\begin{equation*}
  \bfpi(\x)
  \;=\;
  \frac{\lambda^{(k-1)x_k + (k-2)x_{k-1}+ \dots + x_2}}{
    (1 + \lambda + \lambda^2 + \dots + \lambda^{k-1})^m
  }
  \cdot
  \binom{m}{x_1, x_2, \dots, x_k}
\end{equation*}
for all $\x = (x_1, \dots, x_k) \in \Delta^k_m$,
and where $\lambda = a/b$. 

This culminates in Theorem~\ref{thm:ehrenfestkd:stationary}
(introduced in Section~\ref{sec:tech-overview:ehrenfest}),
which we restate here for convenience:

\ehrenfestkdstationary*

\begin{proof}
  We verify that $\bfpi$ satisfies the detailed balance equations
  \begin{equation}
    \pi(\x) \cdot P(\x, \y)
    \;=\;
    \pi(\y) \cdot P(\y, \x)
    \label{eq:dbe}
  \end{equation}
  for all $\x, \y \in \Delta^m_k$.
  Given the description of the non-zero transition probabilities
  of $\P$ from Definition~\ref{def:ehrenfestkd}
  (in particular, defined in terms of the
  variables $p_\x^{j, j+1}$ and $p_\x^{j+1, j}$
  for any $\x \in \Delta^m_k$ and $j \in [k-1]$),
  it suffices by symmetry to
  verify expression~\eqref{eq:dbe} only for transitions between
  pairs of states
  \begin{equation*}
  (x_1, \dots, x_j, x_{j+1}, \dots, x_k)
  \;\;\;\text{and}\;\;\;
  (x_1, \dots, x_j - 1, x_{j+1} + 1, \dots, x_k) \\
  \end{equation*}
  for $j = 0, \dots, k-1$. In other words,
  letting $\x := (x_1, \dots, x_k)$, we wish to verify that
  \begin{equation}
    \bfpi\big(  (x_1, \dots, x_j, x_{j+1}, \dots, x_k) \big)
    \cdot
    p_\x^{j,j+1}
    \;=\;
    \bfpi\big( (x_1, \dots, x_j - 1, x_{j+1} + 1, \dots, x_k)\big)
    \cdot p_\x^{j+1, j}\;.
    \label{eq:dbe-j:1}
  \end{equation}
  For this, first recall in the expression of $\pi$ from
  the theorem statement that the multinomial coefficient is given by
  \begin{equation*}
    \binom{m}{x_1, \dots, x_k}
    \;=\;
    \frac{m!}{x_1 ! \cdots x_k!} \;.
  \end{equation*}
  Thus we can cancel out all matching terms to show that 
  verifying expression~\eqref{eq:dbe-j:1} reduces to verifying
  \begin{equation}
    \frac{\lambda^{(j\cdot x_{j+1} + (j-1)\cdot x_j)}}{
      x_j ! \cdot x_{j+1} !
    }
    \cdot
    p_\x^{j,j+1}
    \;=\;
    \frac{\lambda^{(j\cdot(x_{j+1} + 1) + (j-1) \cdot (x_j - 1))}}{
      (x_{j+1} + 1)! \cdot  (x_j - 1)!
    }
    \cdot
    p_\x^{j+1, j}\;.
    \label{eq:dbe-j:2}
  \end{equation}
  For the left-hand side of~\eqref{eq:dbe-j:2}, we use the
  definition of $p_\x^{j, j+1}$ and the fact that $\lambda = a/b$
  to simplify and write
  \begin{equation*}
    (\text{LHS  of~\eqref{eq:dbe-j:2}})
    \;=\;
      \frac{\lambda^{(j\cdot x_{j+1} + (j-1)\cdot x_j)}}{
      x_j ! \cdot  x_{j+1} !
      }
      \cdot \frac{a \cdot x_j}{m}
      \;=\;
      \frac{a^{(j \cdot x_{j+1} + (j-1)\cdot x_j + 1)}}{
      b^{(j\cdot x_{j+1} + (j-1)\cdot x_j)}}
      \cdot
      \frac{1}{m \cdot (x_j - 1)!  \cdot (x_{j+1})!} \;.
  \end{equation*}
  Here, the final inequality uses the fact that
  $x_j  / (x_j! \cdot x_{j+1}!)  =  1 / ((x_j - 1)!\cdot x_{j+1}!)$.

  For the right-hand side of~\eqref{eq:dbe-j:2}, we can
  similarly simplify and write
  \begin{align*}
    (\text{RHS  of~\eqref{eq:dbe-j:2}})
    &\;=\;
      \frac{\lambda^{(j\cdot(x_{j+1} + 1) + (j-1) \cdot (x_j - 1))}}{
      (x_{j+1} + 1)! \cdot  (x_j - 1)!
      }
      \cdot
      \frac{b \cdot (x_{j+1} + 1)}{m} \\
    &\;=\;
      \frac{\lambda^{(j\cdot x_{j+1} + (j-1) \cdot x_j +1)}}{
      x_{j+1}! \cdot (x_{j} - 1)!}
      \cdot
      \frac{b}{m} 
    \;=\;
      \frac{a^{(j\cdot x_{j+1} +(j-1)\cdot x_j + 1)}}{
      b^{(j \cdot x_{j+1} + (j-1)\cdot x_j)}}
      \cdot
      \frac{1}{m\cdot  x_{j+1}! \cdot (x_{j} - 1)!}\;.
  \end{align*}
  Observing that the left and right-hand sides of
  expression~\eqref{eq:dbe-j:2} for any $j = 0, \dots, k-1$
  thus establishes that $\pi$ is the unique
  stationary distribution for the process.

  Now for each $j \in [k]$ we can define
  \begin{equation*}
    p_j
    \;:=\;
    \frac{\lambda^{(j-1)}}{\sum_{j=1}^k \lambda^{(j-1)}} \;.
  \end{equation*}
  Clearly $\sum_{i \in [k]} p_i = 1$, and we can rewrite
  the PMF of $\bfpi$ as
  \begin{equation*}
    \bfpi(\x) \;=\;
    p_1^{x_1} \cdots p_k^{x_k} \cdot \binom{m}{x_1, \dots, x_k} \;,
  \end{equation*}
  for any $\x = (x_1, \dots, x_k) \in \Delta_k^m$.
  In other words, $\bfpi$ is a \textit{multinomial distribution}
  with parameters $m$ and $(p_1, \dots, p_k)$. 
\end{proof}

\subsection{Bounds on mixing times}
\label{app:ehrenfest:mixing}

In this section, we develop the proof of
Theorem~\ref{thm:ehrenfestkd:mixing},
which bounds the mixing time of the
$(k, a, b, m)$-Ehrenfest process.
The theorem is restated for convenience:

\ehrenfestkdmixing*

In the following subsections, 
we develop the proofs of the upper and lower bounds
on the mixing time separately.

\subsubsection{Upper bounds on mixing time}
\label{app:ehrenfest:mixing:upper}

\paragraph{Coupling setup}
To prove the upper bound on $\tmix$, recall from the
overview in Section~\ref{sec:tech-overview:ehrenfest}
that we introduce a coupling $\{(X_t, Y_t)\}$
over the space
$\Omega \times \Omega = \{1, \dots, k\}^m \times \{1, \dots, k\}^m$,
where the transition probabilities over
the vectors of \textit{counts} of each element
$j \in [k]$ for each of $\{X_t\}$ and $\{Y_t\}$
is a  $(k, a, b, m)$-Ehrenfest process.
We more formally specify the details
of this coupling:
\begin{enumerate}[
  label=-,
  leftmargin=2em,
  rightmargin=3em,
  itemsep=0pt]
\item
  At time $t=0$, assume
  $X_0 = x$ and $Y_0 = y$, for some $x, y \in \Omega$.
\item
  At each time $t \ge 1$, sample $i \in [m]$ uniformly
  at random. 
\item
  Letting $X^i_t$ and $Y^i_t$ denote the
  $i$'th coordinates of $X_t$ and $Y_t$, respectively, set:
  \begin{equation}
    \big(X^i_{t+1}, Y^i_{t+1}\big)
    \;=\;
    \begin{aligned}
      \begin{cases}
        \big([X^i_t + 1], [Y^i_t + 1] \big)
        \;\;\;\text{with probability $a$} \\
        \big([X^i_t - 1], [Y^i_t - 1] \big)
        \;\;\;\text{with probability $b$} \\
        \big(X^i_t, Y^i_t\big)
        \;\;\;\text{otherwise} \\
      \end{cases}
    \end{aligned}
    \;.
    \label{eq:coupling-def-copy}
\end{equation}  
\end{enumerate}

For $j \in \{1, \dots, k\}$,
let $x^t_j$ and $y^t_j$ denote the number of coordinates
in $X_t$ (resp., $Y_t$), where $X^i_t = j$ (resp., $Y^i_t = j$).
Then letting $\x^t = (x^t_1, \dots , x^t_k)$
and $\y^t = (y^t_k, \dots , y^t_k)$, 
it is easy to see that $\{\x^t\}$ and $\{\y^t\}$
are both $(k, a, b, m)$-Ehrenfest processes
under the randomness of the coupling, 
and thus the mixing time of $\{\x^t\}$ (resp., $\{\y^t\}$)
is equivalent to that of $\{X_t\}$ (resp., $\{Y_t\}$).\footnote{
 Note the difference in the time-indexing location 
 between the processes $\{X_t\}$,
 $\{Y_t\}$, and $\{\x^t\}$, $\{\y^t\}$. 
}

Then initialized at $X_0 = x$ and $Y_0 = y$ for
$x, y \in \Omega$,
we let $\taucouple$ denote the \textit{coupling time}
of the process, which is the first 
time $t \ge 0$ such that $X^i_t = Y^i_t$
for all coordinates $i \in [m]$. Formally:
\begin{equation*}
  \taucouple
  \;=\;
  \min \{
  t \ge 0 \;:\; X_s = Y_s \text{ for all $s \ge t$}
  \} \;,
\end{equation*}
and we use the standard fact~\cite[Corollary 5.5]{levin2017markov}
that
\begin{equation}
  d(t)
  \;\le\;
  \max_{x, y \in \Omega} \;
  \Pr\big(\taucouple > t \big) \;.
  \label{eq:coupling-lemma}
\end{equation}
Thus our goal is to derive tail bounds on $\taucouple$,
and for this, we express $\taucouple$
in terms of its \textit{coordinate-wise}
coalescing times: specifically we define
\begin{equation*}
  \taucouple^i
  \;=\;
  \min\{
  t \ge 0 \;:\; X^i_s = Y^i_s \text{ for all $s \ge t$}
  \}
\end{equation*}
for each $i \in [m]$,
and it follows that we can write
$
\taucouple
=
\max_{i \in [m]}\; \taucouple^i
$.
Thus our strategy in proving
Theorem~\ref{thm:ehrenfestkd:mixing} is to first
bound each $\taucouple^i$ in expectation,
and to then use standard machinery to derive a
tail bound on $\taucouple$.
We proceed now to develop these steps.

\paragraph{Bounding $\taucouple^i$ in expectation}

Observe from expression~\eqref{eq:coupling-def-copy}
that conditioned on the set of steps
during which coordinate $i \in [m]$ is selected,
$\{X^i_t\}$ and $\{Y^i_t\}$ each
behave as a biased random walk (with reflecting barriers)
on $\{1, \dots, k\}$.
In particular, given the shared randomness
in $\{(X^i_t, Y^i_t)\}$, observe that the
distance $|X^i_t - Y^i_t|$ is non-increasing with $t$,
and thus the two coordinate-wise random walks must
coalesce at either $1$ or $k$.

Thus for each coordinate $i \in [m]$,
we let $\taucoal^i$ denote the number of times
coordinate $i$ must be sampled in the coupling
of $\{(X_t, Y_t)\}$ until $X^i_t = Y^i_t$. 
Then letting $S^i_j$ denote the waiting time between
when  coordinate $i$ is sampled for the $(j-1)$'th 
and $j$'th times in the coupling $(X_t, Y_t)$, it follows that
\begin{equation}
  \E\big[
  \taucouple^i
  \big]
  \;=\;
  \E\Big[
  {\textstyle \sum_{j=1}^{\taucoal^i}}
  \; S^i_j
  \Big]
  \;=\;
  \E[\taucoal^i] \cdot m  \;,
  \label{eq:tau-couple-i}
\end{equation}
where in the final equality we use the
fact that each $S^i_j$ is an independent geometric
random variable with parameter $1/m$. 
By proving the following uniform bound on $\E[\taucoal^i]$,
we then easily obtain a bound on $\E[\taucouple^i]$
as a consequence of expression~\eqref{eq:tau-couple-i}:

\begin{lemma}
  \label{lem:taucoal-i}
  For any $i \in [m]$, consider the coordinate-wise
  coupling $\{(X^i_t, Y^i_t)\}$ on $\{1, \dots, k\}$
  with parameters $a, b > 0$ where $a + b \le 1$.
  Then
  \begin{equation*}
    \E\big[\taucoal^i\big]
    \;\le\;
    \begin{aligned}
      \begin{cases}
        \;\;\min\Big\{
        \frac{k}{|a-b|}, k^2 
        \Big\}
        \;\;\;&\text{when $a \neq b$} \\
        \;\;k^2 
        \;\;\;&\text{when $a=b$}
      \end{cases}
      \;\;.
    \end{aligned}
  \end{equation*}
\end{lemma}

The proof of Lemma~\ref{lem:taucoal-i} follows
by reducing the coalescing time of two
coupled, biased random walks on $\{1, \dots, k\}$
to the absorbption time of a single, biased
random walk on $\{-k, \dots, k\}$.
For this, assume without loss of generality that
$X^i_0 = z$ and $Y^i_0 = w$ for $1 \le z < w \le k$.
For simplicity, we assume all time indices henceforth
are conditioned on coordinate $i$ being selected
in the parent, $m$-coordinate coupling.
Now we define a third random walk $\{Z_t\}$ on $\{-k, \dots, k\}$
that increments and decrements with probability
$a$ and $b$ using the same shared randomness of $(X^i_t, Y^i_t)$.
Specifically, we have:
\begin{equation}
  \big( X^i_{t+1}, Y^i_{t+1}, Z_{t+1} \big)
  \;=\;
  \begin{aligned}
    \begin{cases}
      \big([X^i_t + 1], [Y^i_t + 1], Z_t + 1 \big)
      \;\;\;\text{with probability $a$} \\
      \big([X^i_t - 1], [Y^i_t - 1], Z_t - 1\big)
      \;\;\;\text{with probability $b$} \\
      \big(X^i_t, Y^i_t, Z^i_t \big)
      \;\;\;\text{otherwise} \\
    \end{cases}
  \end{aligned}
  \;.
  \label{eq:three-way-coupling}
\end{equation}
Now initialize $Z_t = z$, and let $\tauabsorb$
denote the first time $Z_t$ is absorbed at either $-k$ or $k$, i.e.:
\begin{equation*}
  \tauabsorb \;=\;
  \min \big\{ t \ge 0 \;:\; Z_t \in \{-k, k\} \big\} \;.
\end{equation*}
In the following proposition, we show that $\taucoal^i \le \tauabsorb$.

\begin{proposition}
  \label{prop:tau-relationship}
  For any $k \ge 2$, and any $i \in [m]$, consider the
  coupling $\{(X^i_t, Y^i_t, Z_t)\}$ as
  defined in expression~\eqref{eq:three-way-coupling}
  with parameters $a, b > 0$ and $a + b \le 1$.
  Assume that $X^i_0 = z < w = Y^i_0$, for
  some $z, w \in [k]$ and that $Z_0 = 0 < z$.
  Then $\taucoal^i \le \tauabsorb$. 
\end{proposition}

\begin{proof}
  Observe from the transition probabilities
  of $\{X^i_t\}$ and $\{Y^i_t\}$ that
  at time $\taucoal^i$, we must have
  $X^i_{\taucoal} = Y^i_{\taucoal} \in \{1, k\}$.
  Notice also by the initialization choice of $\{Z_t\}$
  that $Z_t < X^i_t \le Y^i_t$ for all $t \ge 0$.
  Thus if $Z_{\tauabsorb} = k$,
  then it must be the case that
  $\tauabsorb \ge \taucoal^i$, and that $X^i_{t} = k$
  at time $t = \taucoal^i$. 
  On the other hand, given that the randomness
  between the three processes is shared,
  observe that if $Z_t = j < 0$ for some
  $t \ge 0$, then $Y_t - X_t \le (w-z) - |j|$.
  This is due to the fact that starting from 
  $Z_t = 0$, every instance of $Z_t$ decreasing to
  a smaller value $j$ for the first time 
  corresponds to the process $X^i_t$ remaining at 1,
  and the process $Y^i_t$ decrementing by 1.
  Thus if $Z_{\tauabsorb} = -k$, then we must have
  \begin{equation*}
    Y^i_{\tauabsorb} - X^i_{\tauabsorb}
    \;\le\; (w-z) - k
    \;\le\; 0 \;,
  \end{equation*}
  where the final inequality follows from the fact
  that $w-z \le k$.
  Since $1 \le X^i_t \le Y^i_t$ for all $t$ by definition,
  this implies that 
  $Y^i_{\tauabsorb} = X^i_{\tauabsorb} = 1$.
  Thus if $Z_{\tauabsorb} = -k$,  we again conclude
  that $\tauabsorb > \taucoal^i$. 
\end{proof}

Given Proposition~\ref{prop:tau-relationship}, we now
proceed to bound $\E[\tauabsorb]$ for the process
$\{Z_t\}$. In turn, this yields the desired inequality
for $\E[\taucoal^i]$ from Lemma~\ref{lem:taucoal-i}.

\begin{proposition}
  \label{prop:tau-absorb}
  Fix $a , b > 0$ with $a + b \le 1$. 
  Let $\{Z_t\}$ be the random walk
  on $\{-k, \dots, k\}$ that increments and decrements
  at each step with probabilities $a$ and $b$ respectively.
  Let $Z_0 = 0$, and let $\tauabsorb$ denote the
  first time $Z_t \in \{-k, k\}$.
  Then
  \begin{equation*}
    \E\big[ \tauabsorb \big]
    \;\le\;
    \begin{aligned}
      \begin{cases}
        \;\;\min\Big\{
        \frac{k}{|a-b|}, k^2 
        \Big\}
        \;\;\;&\text{when $a \neq b$} \\
        \;\;k^2 
      \;\;\;&\text{when $a=b$}
    \end{cases}
  \end{aligned}
  \;\;\;.
  \end{equation*}
\end{proposition}

\begin{proof}
  We start with the case where $a \neq b$, and we assume
  without loss of generality that $a > b$. 
  We apply a standard martingale argument used to
  derive expected stopping times for
  gambler's ruin-type random walks~\cite{grimmett2020probability}.
  For completeness, we provide the full argument,
  and to start we define the two processes:
  \begin{equation*}
    U_t = Z_t - (a-b)t \;.
    \;\;\;\text{and}\;\;\;
    M_t = (b/a)^{Z_t} \;.
  \end{equation*}
  We can verify that both $\{U_t\}$ and $\{M_t\}$ are martingales
  with respect to $\{Z_t\}$ by computing:
  \begin{align*}
    \E\big[U_{t+1} \;|\; \{Z_t\} \big]
    &\;=\;
      a (Z_t + 1) + b(Z_t - 1) + (1-a-b)Z_t - (a-b)(t+1) \\
    &\;=\;
      Z_t(a + b + (1-a-b)) + (a-b)(1 - (t+1)) \\
    &\;=\;
    Z_t + (a-b) t  \;=\; U_t\;,
  \end{align*}
  and
  \begin{align*}
    \E\big[M_{t+1} \;|\; \{Z_t\} \big]
    &\;=\;
      a(b/a)^{Z_t + 1} +
      b(b/a)^{Z_t - 1} +
      (1-a-b)(b/a)^{Z_t} \\
    &\;=\;
      a(b/a)(b/a)^{Z_t}
      + b (a/b)(b/a)^{Z_t}
      + (1-a-b)(b/a)^{Z_t} \\
    &\;=\;
      (b/a)^{Z_t} \cdot
      (b + a + (1-a-b)) \\
    &\;=\;
      (b/a)^{Z_t}
      \;=\;
      M_t\;.
  \end{align*}
  Now observe that $\tauabsorb$ is a stopping
  time for $\{Z_t\}$, and that $\E[\tauabsorb]$ is finite
  given that the probability that $Z_t$ is absorbed
  at $k$ in the next $2k$ steps is at least $a^{2k}$.
  Moreover, the increments of both $\{U_t\}$
  and $\{M_t\}$ are bounded by absolute constants.
  Thus we can apply the Optional Stopping Theorem
  \cite[Theorem 12.5.2]{grimmett2020probability} to
  both processes.
  For the latter process, letting $p^+$ denote the
  probability that $Z_{\tauabsorb} = k$ and letting
  $\lambda := a/b > 1$, this implies that
  \begin{equation*}
    \E\big[ M_{\tauabsorb} \big]
    \;=\;
    p^+ (1/\lambda)^k + (1-p^+)(1/\lambda)^{-k}
    \;=\;
    \E\big[Z_0 \big]
    \;=\;
    (1/\lambda)^0
    \;=\;1
    \;.
  \end{equation*}
  Solving for $p^+$ and rearranging, we find
  \begin{equation}
    p^+
    \;=\;
    \frac{1- (1/\lambda)^{-k}}{(1/\lambda)^k - (1/\lambda)^{-k}}
    \;=\;
    \frac{\lambda^k - 1}{\lambda^k - (1/\lambda)^k} \;.
    \label{eq:p-plus}
  \end{equation}
  Applying the Optional Stopping Theorem
  to the former process $\{U_t\}$ shows
  \begin{equation*}
    \E\big[ U_{\tauabsorb} \big]
    \;=\;
    \E\big[Z_{\tauabsorb} - (a-b)\tauabsorb\big]
    \;=\;
    \E\big[U_0\big]
    \;=\; 0 \;,
  \end{equation*}
  which implies that $\E[\tauabsorb] = \E[Z_{\tauabsorb}] / (a-b)$.
  Then using the fact that
  \begin{equation*}
    \E\big[ Z_{\tauabsorb} \big]
    \;=\;
    p^+ k - (1-p^+)k
    \;=\;
    k (2 p^+ - 1)  
  \end{equation*}
  and the expression for $p^+$ from~\eqref{eq:p-plus},
  we conclude that
  \begin{equation}
    \E[Z_{\tauabsorb}]
    \;=\;
    \frac{k \cdot (2 p^+ - 1)}{a-b}
    \;=\;
    \frac{k}{a-b} \cdot
    \bigg(
    \frac{2\big(\lambda^k -1 \big)}{
      \lambda^k - (1/\lambda)^k
    }
    - 1
    \bigg)
    \;.
    \label{eq:tau-absorb-general}
  \end{equation}
  In general, $2p^+ - 1 \le 1$, which means
  $\E[Z_{\tauabsorb}] \le \frac{k}{a-b}$.
  On the other hand, this yields a loose bound
  when $|a-b| = o(1/k)$.
  In this case, it can be verified that
  the right hand side of~\eqref{eq:tau-absorb-general}
  is bounded by $k^2$ (for all $k \ge 2$)
  for any $a > b$. Thus in this regime, we conclude that
  $\E[\taucouple^i] \le \min\{\frac{k}{a-b}, k^2 \}$,
  and it follows by symmetry that
  $\E[\taucouple^i] = \min\{\frac{k}{b-a}, k^2 \}$,
  when $b > a$.

  For the case when $a=b$, we can apply a similar
  argument to the martingale $\{G_t\}$ where
  $G_t = (Z_t)^2 - t$ (i.e., the standard
  analysis for an unbiased gambler's ruin random walk
  ~\cite{grimmett2020probability})
  to show in this regime that $\E[\taucouple^i] = k^2$,
  which concludes the proof. 
\end{proof}

The proof of Lemma~\ref{lem:taucoal-i}
then follows from Propositions~\ref{prop:tau-relationship}
and~\ref{prop:tau-absorb},
and together with expression~\eqref{eq:tau-couple-i}
this implies that
\begin{equation}
  \E\big[ \taucouple^i \big]
  \;\le\;
  \begin{aligned}
    \begin{cases}
      \;\;\min\Big\{
      \frac{k}{|a-b|}, k^2 
      \Big\}
      \cdot m
      \;\;\;&\text{when $a \neq b$} \\
      \;\;k^2 \cdot m
      \;\;\;&\text{when $a=b$}
    \end{cases}
  \end{aligned}
  \;.
  \label{eq:tau-couple-i-bound}
\end{equation}

\paragraph{Tail bound on $\taucouple$} \
Recall that $\taucouple$ is defined as the coupling
time of the process $\{(X_t, Y_t)\}$, and we can write
\begin{equation*}
  \taucouple
  \;=\;
  \max\nolimits_{i \in [m]}
  \taucouple^i \;.
\end{equation*}
Using the bound on $\E[\taucouple^i]$ from
expression~\eqref{eq:tau-couple-i-bound},
we now prove the following bound on the tail of
$\taucouple$. As a consequence of the
coupling property from expression~\ref{eq:coupling-lemma},
the statement of Theorem~\ref{thm:ehrenfestkd:mixing}
will immediately follow.

\begin{lemma}
  \label{lem:taucouple}
  Fix $k, m\ge 2$ and $a > b > 0$ with $a + b \le 1$.
  Consider the resulting process $\{(X_t, Y_t)\}$
  as defined in expression~\eqref{eq:coupling-def-copy},
  intialized at $(X_0, Y_0) = (x, y)$ for
  $x, y \in \{1, \dots, k\}^m$.
  Then letting
  \begin{equation*}
    \Phi
    \;:=\;
    \begin{aligned}
      \begin{cases}
        \;\min\big\{\frac{k}{|a-b|}, k^2 \big\} \cdot m
        \;\;\;&\text{when $a \neq b$} \\
        \;k^2 \cdot m
        \;\;\;&\text{when $a = b$}
      \end{cases}
    \end{aligned}
    \;\;,
  \end{equation*}
  it follows that
  $\Pr\big( \taucouple > 2 \Phi \cdot \log(4m) \big) \le \frac{1}{4}$. 
\end{lemma}

\begin{proof}
  To start, we prove by induction that for
  every $i \in [m]$, and $c \ge 1$:
  \begin{equation}
    \Pr\Big(
    \taucouple^i > 
    c \cdot 2 \Phi
    \Big)
    \;\le\;
    \frac{1}{2^c} \;.
    \label{eq:coupling-induction}
  \end{equation}
  Consider the base case when $c = 1$. 
  It follows by Markov's inequality that
  \begin{equation}
    \Pr\Big(
    \taucouple^i > 
    2 \Phi
    \Big)
    \;\le\;
    \frac{\E\big[ \taucouple^i \big]}{2 \Phi}
    \;\le\;
    \frac{1}{2}
    \;,
    \label{eq:coupling-base}
  \end{equation}
  where in the final inequality we
  used the bound $\E[\taucouple^i] \le \Phi$
  from expression~\eqref{eq:tau-couple-i-bound}.
  Now assume that the claim holds for some $c - 1 \ge 1$. 
  Then we can write
  \begin{align*}
    \Pr\Big(
    \taucouple^i  >
    c \cdot 2 \Phi
    \Big)
    &\;=\;
      \Pr\Big(
      \taucouple^i  >
      c \cdot 2 \Phi \;|\;
      \taucouple^i > (c-1) \cdot 2 \Phi
      \Big)
      \cdot
      \Pr\Big(
      \taucouple^i > (c-1) \cdot 2 \Phi
      \Big) \\
    &\;=\;
      \Pr\Big(
      \taucouple^i  > 2 \Phi
      \Big)
      \cdot
      \Pr\Big(
      \taucouple^i > (c-1) \cdot 2 \Phi
      \Big) \\
    &\;\le\;
      \frac{1}{2} \cdot \frac{1}{2^{(c-1)}}
      \;=\; \frac{1}{2^c}
    \;,
  \end{align*}
  where the equality in the second line follows
  from the independence of each step
  in the coupling, and where the final inequality
  comes from applying the bound from
  expression~\eqref{eq:coupling-base}
  and by the inductive hypothesis.
  Thus we conclude that expression~\eqref{eq:coupling-induction}
  holds in general for all $c \ge 1$.

  It follows by a union bound that
  \begin{equation*}
    \Pr\Big(
    \taucouple > c \cdot 2 \Phi
    \Big)
    \;=\;
    \Pr\Big(
    \max_{i \in [m]}\; \taucouple^i > c \cdot 2 \Phi
    \Big)
    \;\le\;
    \sum_{i \in [m]}
    \Pr\Big(
    \taucouple^i > c \cdot 2 \Phi
    \Big)
    \;\le\;
    \frac{m}{2^c} \;.
  \end{equation*}
  Then setting $c := \log (4m)$ implies
  $\Pr\big(\taucouple > 2 \Phi \cdot \log (4m) \big) \le \frac{1}{4}$,
  which concludes the proof. \qedhere
\end{proof}

Given that the bound in Lemma~\ref{lem:taucouple}
is uniform over all $x, y \in \{1, \dots, k\}^m$, it follows
from expression~\eqref{eq:coupling-lemma} that
\begin{equation*}
  d\big( 2 \Phi \cdot \log (4m)\big)
  \;\le\;
  \frac{1}{4} \;,
\end{equation*}
where we use the definition of $\Phi$ from the lemma.
By definition, this means
 $\tmix = O\big( \Phi \cdot \log m \big)$,
which concludes the proof of Theorem~\ref{thm:ehrenfestkd:mixing}.

\subsubsection{Lower bounds on mixing times}
\label{app:ehrenfest:mixing:lower}

We now prove the lower bound on $\tmix$ from 
Theorem~\ref{thm:ehrenfestkd:mixing}, which states that 
for the ($k, a, b, m$)-Ehrenfest process,
$\tmix \ge \Omega(km)$. 
This bound follows from a standard
diameter lower bound argument (e.g., of~\cite{levin2017markov}),
which is constructed as follows:
\begin{enumerate}[
  label=-,
  leftmargin=2em,
]
\item
  Given the transition matrix $\P$ of the $(k, a,b,m)$-Ehrenfest
  process over the state space $\Delta^m_k$, 
  let be the $G$ be the graph whose vertex set 
  is $\Delta^m_k$, and whose edge set $E$ is given by
  \begin{equation*}
    E = \{ (\x, \y) \in \Delta^m_k \times \Delta^m_k  \;:\;
    \P(\x, \y) + \P(\y, \x) > 0 \} \;.
  \end{equation*}
\item 
  Let $D$ be the diameter of the graph $G$. 
  Then it follows that
  $\tmix \ge \frac{D}{2}$ \cite[Section 7.1.2]{levin2017markov}. 
\end{enumerate}

Using this construction, we obtain a lower bound on
$\tmix$ by deriving the following  lower bound on the 
diameter $D$. 

\begin{proposition}
  \label{prop:ehrenfest-diameter-lb}
  Fix $a, b > 0$ with $a + b \le 1$ and $k, m \ge 2$,
  and let $\tmix$ be the mixing time of the
  $(k, a, b, m)$-Ehrenfest process.
  Consider the construction of the graph $G$ defined above,
  and let $D$ denote its diameter. 
  Then $D \ge km$, and thus $\tmix \ge \Omega(km)$. 
\end{proposition}

\begin{proof}
  Recall that the diameter $D$ is defined 
  as the maximal shortest-path distance 
  (counter by the number of edges) 
  between two distinct vertices in $G$.
  Then consider the vertices 
  $\x = (m, 0, \dots, 0) \in \Delta^m_k$ and 
  $\y = (0, \dots, 0, m) \in \Delta^m_k$.
  It is clear from the transition probabilities in 
  Definition~\ref{def:ehrenfestkd} 
  that $\x$ and $\y$ are connected in $G$. 
  Moreover, for any $a$ and $b$,
  any path between $\x$ and $\y$ in $G$  
  is comprised of at least $km$ edges.
  Note that this is due to the fact that
  every non-null transition in $\P$ can 
  increments the count of at most coordinate
  of $\z \in \Delta^k_m$ by at most 1. 
  Thus at least $km$ total transitions are needed
  to transport the full weight $m$ in coordinate
  1 of $\x$ to the last coordinate $k$.
  It follows that $D \ge km$, 
  and that $\tmix \ge km/2$
  by the diameter lower bound.
\end{proof}

Thus Proposition~\ref{prop:ehrenfest-diameter-lb}
establishes the mixing time lower bound
from Theorem~\ref{thm:ehrenfestkd:mixing}.

\section{Details on Repeated Donation Game Payoffs}
\label{app:gt}

In this section, we derive exact expressions
for the expected payoff functions $f(\cdot, \cdot)$
in RD games.
We then develop the proof of Proposition~\ref{prop:params},
which establishes the local optimality of
the $k$-IGT dynamic update rules under certain
parameter regimes.

\subsection{Computing expected payoffs for $\G$ agents}
\label{app:gt:expected-payoffs}

Recall from Section~\ref{sec:intro:RDG}
that for two strategies $\S_1, \S_2 \in \{\AC, \AD, g_1, \dots, g_k\}$,
we define $f(\S_1, \S_2)$ as the expected payoff of a
node playing strategy $\S_1$ against an opponent
playing strategy $\S_2$ during a single RD game.
We proceed to make this definition precise:

\subsubsection{Defining the expected payoff function}
\label{app:gt:expected-payoffs:f-def}

We first recall the necessary components which
were introduced in Sections~\ref{sec:intro:RDG}:
\begin{enumerate}[
  label=-,
  itemsep=0.5em,
  leftmargin=2em,
  rightmargin=3em,
  ]
\item
  RD games have four game states
  $\calA = (CC, CD, DC, DD)$.
  Each state is an ordered pair
  specifying the action of a row player
  with strategy $\S_1$ and a column player
  with strategy $\S_2$ during a given round. 
\item
  In RD games, a strategy is comprised of an initial distribution
  over the actions $(C, D)$, and a (randomized)
  transition rule for choosing an action in
  the next round (conditioned on playing an additional
  round with probability $\delta$).
\item
  For the pair of strategies $(\S_1, \S_2)$,
  let $\M$ be the row-stochastic matrix specifying
  the transition probabilities over $\calA$
  \textit{conditioned on playing an additional round of the game.}
  We assume that the rows and columns of $\M$ are indexed
  by the four states in $\calA$, and thus
  $\M_{1, 1}$ denotes the (conditional) probability
  of transitioning from state $CC$ to $CC$,
  $\M_{1, 2}$ denotes the (conditional) probability
  of transitioning from state $CC$ to $CD$, etc. 
    
\item
  Let $\q_1$ denote the initial
  distribution over the game states $\calA$
  determined by the initial action distributions
  of the two strategies $\S_1$ and $\S_2$.
  For $i \ge 2$,
  let $\q_i \in \Delta_4$ denote the distribution
  over game states $\calA$
  conditioned on (i) having already played $i-1$
  rounds of the game, and on (ii) playing
  an additional round with probability $\delta$.
  It follows that
  \begin{equation*}
    \begin{aligned}
      \q_2 &\;=\; \q_1 \M \\
      \text{and}\;\;
      \q_i &\;=\; \q_{i-1} \M 
      \;=\; \q_1 \M^{i-1}
      \;\;\;\text{for all $i \ge 3$}  \;.
    \end{aligned}
  \end{equation*}
  
\item
  Given the single-round donation game payoff matrix
  \begin{align}
    \vspace*{-0.5em}
    \small
    \begin{tabular}{c|cc}
      & $C$ & $D$ \\
      \hline
      $C$ & $b-c$ & $-c$ \\
      $D$ & $b$ & $0$
    \end{tabular}
                  \quad\quad\text{where $b > c \ge 0$}
                  \label{eq:pd-payoff-matrix-copy}
                  \;,
  \end{align}
  let $\v := [b-c, -c, b, 0]^\top$ be the vector of 
  single-round (row player) payoffs.
\end{enumerate}

Then in the RD game setting with restart probability $\delta$,
the pair of strategies $(\S_1, \S_2)$ specify a Markov
chain $\{\wq_i\}$ over the joint space
$\hatcalA := (\calA, \bot)$,
where $\bot$ denotes the termination state of
the repeated game (that is reached with probability
$(1-\delta)$ at the end of each round).
In particular, using the components above,
we can define each $\wq_i \in \Delta_5$ by:
\begin{equation}
  \begin{aligned}
    \wq_1
    &\;=\;
    \big[\;\q_1  \;\;\; 0\;\big] \\
    \wq_2
    &\;=\;
    \big[\;\delta \q_2 \;\;\; (1-\delta)\;\big]
    \;=\;
    \big[\;\delta \q_1 \M \;\;\; (1-\delta) \;\big]\\
    \text{and}\;\;
    \wq_i
    &\;=\; 
    \big[\;\delta^{i-1} \q_{i-1} \;\;\; (1-\delta^{i-1}) \;\big]
    \;=\;
    \big[\; \q_1 (\delta \M)^{i-1} \;\;\; (1-\delta^{i-1}) \;\big]
    \;\;\;\text{for all $i \ge 3$}  \;\;.
  \end{aligned}
  \label{eq:wqi-def}
\end{equation}
Now let $\wv := [\; \v \;\; 0 \;]$ denote the
vector of single-round (row player) payoffs over the
repeated game states $\hatcalA$
(where players receive a payoff of 0 when the
game terminates and enters state $\bot$).
For all $i \ge 1$, we define
$r_i$ as the payoff of the row player
(with strategy $\S_1$) in round $i$.
Then given the pair of strategies $(\S_1, S_2)$
the \textit{expected} reward at round $i \ge 1$
(over the randomness of the
strategies and the repeated game probability)
is given by
\begin{equation}
  \E[r_i]
  \;=\;
  \big\langle
  \wv,  \wq_i
  \big\rangle
  \;=\;
  \big\langle
  \v, \delta^{i-1} \q_{i-1}
  \big\rangle
  \;=\;
  \big\langle
  \v,
  \q_1 (\delta \M)^{i-1}
  \big\rangle
  \;,
  \label{eq:ri-def}
\end{equation}
where in the second equality we use
the fact that that last component of $\wv$ is $0$
and the definition of $\wq_i$
from expression~\eqref{eq:wqi-def}.
Then we can formally define the expected payoff $f(\S_1, \S_2)$
for the node with strategy $\S_1$ over the entire RD game as:
\begin{equation}
  f(\S_1, \S_2)
  \;:=\;
  \sum_{i = 1}^{\infty}\;
  \E[r_i]
  \;=\;
  \sum_{i = 1}^{\infty}\;
  \big\langle
  \v, \q_1 (\delta \M)^{i-1}
  \big\rangle
  \;=\;
  \Big\langle
  \v, \q_1 \Big(\sum_{i=1}^\infty (\delta \M)^{i-1} \Big)
  \Big\rangle \;.
  \label{eq:f-def}
\end{equation}

Using this definition, in the following subsections we
derive the expected payoffs for $\gtft$ agents
against each of the strategy types 
$\AC$, $\AD$, and $\gtft$.
In each case, we define the matrix $\M$
and initial distribution $\q_1$ specified by the
pair of strategies, and we then
compute the value of $f$ using expression~\eqref{eq:f-def}.
We summarize the expected payoff expressions
in Section~\ref{app:gt:expected-payoffs:summary}.

\subsubsection{Expected payoff for $\gtft$ against $\AC$}
\label{app:gt:expected-payoffs:g-c}

For a fixed generosity parameter $g$,
we compute $f(g, \AC)$. 
Recall that in the first round of an RD game,
agents with strategy $\G$ play $C$ with
probability $s_1$, and $D$ with probability $(1-s_0)$.
It follows that the initial distribution $\q_0$
over the game states $\calA = (CC, CD, DC, DD)$  
for the strategy pair $(\G, \AC)$ is given by
\begin{equation}
  \q_1 \;=\;
  \big[s_1, 0, (1-s_1), 0\big]^\top \;.
  \label{eq:q0:g-c}
\end{equation}
Moreover, the transition matrix $\M$ over $\calA$
(conditioned on playing an additional round) is specified by
\begin{equation}
  \M \;=\;
  \begin{pmatrix}
    1 & 0 & 0 & 0  \\
    g & 0 & (1-g) & 0  \\
    1 & 0 & 0 & 0  \\
    g & 0 & (1-g) & 0 
  \end{pmatrix}
  \;\;.
  \label{eq:M:g-c}
\end{equation}
which comes from the fact that the agent with strategy
$\AC$ plays $C$ at each round. 
Recalling that $\v = [b-c, -c, b, 0]^\top$, it follows that
\begin{align*}
  \q_1 (\delta \M)^{i-1}
  &\;=\; [\;\delta^{i-1}, 0, 0, 0\;]^\top \\
  \text{and thus}\;\;
  \big\langle
  \v, \q_1 (\delta \M)^{i-1}
  \big\rangle
  &\;=\;
    (b-c) \cdot \delta^{i-1}
\end{align*}
for all $i \ge 2$.
Then using the definition of $f$
in expression~\eqref{eq:f-def},
we have
\begin{align}
  f(\G, \AC)
  &\;=\;
    \sum_{i=1}^\infty
    \big\langle
    \v, \q_1 (\delta \M)^{i-1}
    \big\rangle
    \nonumber \\
  &\;=\;
    \big\langle \v, \q_1\big\rangle
    +
    \sum_{i=2}^\infty
    \big\langle
    \v, \q_1 (\delta \M)^{i-1}
    \big\rangle
    \nonumber \\
  &\;=\;
    s_1 (b-c) + b(1-s_1) + \frac{\delta \cdot (b-c)}{1-\delta}
    \;=\;
    c(1-s_1) + \frac{b-c}{1-\delta} \;,
    \label{eq:f-def:g-c}
\end{align}
where in the penultimate inequality we use the
fact that $\sum_{i=2}^\infty \delta^{i-1} = \delta/(1-\delta)$. 

\subsubsection{Expected payoff for $\gtft$ against $\AD$}
\label{app:gt:expected-payoffs:g-d}

Using the same approach as in the previous section,
we now compute $f(\G, \AD)$
for a fixed generosity parameter $g$.
Given that agents with strategy $\AD$ play $D$
at all rounds $i \ge 1$, we have 
\begin{equation}
  \q_1 \;=\;
  \big[0, s_1, 0, (1-s_1)\big]^\top \;.
  \label{eq:q0:g-d}
\end{equation}
The transition matrix $\M$ over $\calA$
(conditioned on playing an additional round) in this
case is specified by
\begin{equation}
  \M \;=\;
  \begin{pmatrix}
    0 & 1 & 0 & 0  \\
    0 & g & 0 & (1-g)  \\
    0 & 1 & 0 & 0  \\
    0 & g & 0 & (1-g) 
  \end{pmatrix}
  \;\;.
  \label{eq:M:g-d}
\end{equation}
which again follows directly from the definitions
of the two strategies.
Then for all $i \ge 2$, we have
\begin{align*}
  \q_1 (\delta \M)^{i-1}
  &\;=\;
    \delta^{i-1} \cdot [\; 0, g, 0, (1-g)\;]^\top \\
    \text{and}\;\;
    \big\langle
    \v, \q_1 (\delta \M)^{i-1}
    \big\rangle
  &\;=\;
    \delta^{i-1} \cdot \big(-cg \big)
    \;.
\end{align*}
Again using the definition of $f$ 
from expression~\eqref{eq:f-def}, we can compute
\begin{align}
  f(\G, \AD)
  &\;=\;
    \sum_{i=1}^\infty
    \big\langle
    \v, \q_1 (\delta \M)^{i-1}
    \big\rangle
    \nonumber \\
  &\;=\;
    -cs_1 -
     \frac{cg\delta}{1-\delta} \;.
    \label{eq:f-def:g-d}
\end{align}

\subsubsection{Expected payoff for $\gtft$ against $\gtft$}
  
Given two generosity parameters $g$ and $g'$,
we now derive $f(g, g')$,
which requires more work compared to the previous
two cases. To start, observe that
\begin{equation}
  \q_1 \;=\;
  \big[\;s_1^2, s_1(1-s_1), s_1(1-s_1), (1-s_1)^2\;\big]^\top \;.
  \label{eq:q0:g-g}
\end{equation}
which follows from the assumption that
nodes any $\G$ strategy play the same
initial action parameter $s_1$.
Then by definition of the two strategies,
the transition matrix $\M$ over $\calA$
(conditioned on playing an additional round) is defined by
\begin{equation}
  \M
  \;=\;
  \begin{pmatrix}
    1 & 0 & 0 & 0  \\
    g & 0 & (1-g) & 0  \\
    g' & 1-g' & 0 & 0  \\
    g g' & (1-g')g & g'(1-g) & (1-g)(1-g')
  \end{pmatrix}
  \;\;.
  \label{eq:M:g-g}
\end{equation}
Now recall from the definition of $f$ in
expression~\ref{eq:f-def} that we can write
\begin{equation}
  f(\G, \G')
  \;=\;
  \Big\langle
  \v, \q_1 \Big(\sum_{i=1}^\infty (\delta \M)^{i-1} \Big)
  \Big\rangle \;.
  \label{eq:f-def:g-g:pre}
\end{equation}

Recall for that for a matrix $\B$ whose
eigenvalues are uniformly bounded by 1 in
absolute value, we have the identity
\begin{equation*}
  \sum_{i=1}^\infty \B^{i-1}
  \;=\;
  (\I-\B)^{-1} \;.
\end{equation*}
It follows that
\begin{equation*}
  \sum_{i=1}^\infty
  (\delta \M)^{i-1}
  \;=\;
  (\I - \delta \M)^{-1} \;,
\end{equation*}
and we let $\A := (\I-\delta \M)^{i-1}$
with entries $a_{ij}$ for $i, j \in [4]$.
One can verify that the entries of $\A$ are
given by the following:

{\small
\begin{align*}
  &\A_{1, 1}
    =
    \frac{1}{1 - \delta}
    \;
  &\A_{2, 1}
  &=
    \frac{-\delta^2 g g' + \delta^2 g' + \delta g}{(1 - \delta)(1 - \delta^2(1 - g)(1 - g'))}
  \\
  &\A_{1, 2} = 0
    \;
  &\A_{2, 2}
  &\;=\; \frac{1}{1 - \delta^2 (1 - g)(1 - g')}
  \\
  &\A_{1, 3} = 0
    \;
  &\A_{2, 3}
  &= \frac{\delta - \delta g}{1 - \delta^2 (1 - g)(1 - g')}
  \\
  &\A_{1, 4} = 0
    \;
  &\A_{2, 4}
  &= 0
  \\ \\
  &\A_{3, 1}
    =
    \frac{-\delta^2 g g' + \delta^2 g + \delta g'}{(1 - \delta)(1 - \delta^2(1 - g)(1 - g'))}
    \;
  &\A_{4, 1}
  &= \frac{\delta^2(g g' (\delta (1-g)(1-g') + 1) + g^2(1 - g') + g'^2 (1 - g))}{
    (1-\delta)(1 - \delta(1-g)(1-g'))(1 - \delta^2(1-g)(1-g'))}
  \\
  &\A_{3, 2} = \frac{\delta - \delta g'}{1 - \delta^2 (1 - g)(1 - g')}
    \;
  &\A_{4, 2}
  &= \frac{\delta(\delta g'(1-g)(1-g') + g(1-g'))}{(1 - \delta(1-g)(1-g'))(1 - \delta^2(1-g)(1-g'))}
  \\
  &\A_{3, 3} = \frac{1}{1 - \delta^2 (1 - g)(1 - g')}
    \;
  &\A_{4, 3}
  &= \frac{\delta(\delta g(1-g)(1-g') + g'(1-g))}{(1 - \delta(1-g)(1-g'))(1 - \delta^2(1-g)(1-g'))}
  \\
  &\A_{3, 4} = 0
    \;
  &\A_{4, 4}
  &= \frac{1}{1 - \delta (1 - g)(1 - g')} \;\;.
\end{align*}
}
Then by the definition in expression~\eqref{eq:f-def:g-g:pre}
and using the matrix $\A$, we can write
\begin{equation*}
  f(\G, \G')
  \;=\;
  \big\langle
  \v, \q_1 \A 
  \big\rangle
  \;=\;
  \big\langle
  \v, \q_1 \A 
  \big\rangle
  \;.
\end{equation*}
Recalling that $\v = [b-c, -c, b, 0]^\top$,
and using the definition of $\q_1$ from
expression~\eqref{eq:q0:g-g} and the
entries of $\A$ above, it is straightforward to verify that
\begin{equation}
  \begin{aligned}
    f(\G, \G')
    \;=\;
    s_1(b - c) + \frac{(b-c)\delta}{1 - \delta}
    &+
      c(1 - s_1) \frac{\delta^2(1 - g)(1 - g') + \delta (1 - g)}{1 - \delta^2(1 - g)(1 - g')}
    \\
    &-
      b(1 - s_1) \frac{\delta^2(1 - g)(1 - g') + \delta(1 - g')}{1 - \delta^2(1 - g)(1 - g')}
    \;\;.
  \end{aligned}
  \label{eq:f-def:g-g}
\end{equation}

\subsubsection{Summary of expected $\gtft$ payoffs}
\label{app:gt:expected-payoffs:summary}

We summarize the expected payoffs for agents with
strategy type $\gtft$ derived in the previous three sections.
Summarizing expressions~\eqref{eq:f-def:g-c},
~\eqref{eq:f-def:g-d}, and~\eqref{eq:f-def:g-g}, we have
\begin{align}
  f(g, \AC)
  &\;=\;
  c(1-s_1) + \frac{b-c}{1-\delta} \;.
  \label{f-def:g-c:general} \\
  f(g, \AD)
  &\;=\;
    -cs_1 -
    \frac{cg\delta}{1-\delta} \;.
    \label{f-def:g-d:general}  \\
  f(g, g')
  &\;=\;
    s_1(b - c) + \frac{(b-c)\delta}{1 - \delta}
    + c(1 - s_1) \frac{\delta^2(1 - g)(1 - g') + \delta (1 - g)}{1 - \delta^2(1 - g)(1 - g')}
    \nonumber \\
  &\qquad\qquad\qquad\qquad\quad\;\;
    - b(1 - s_1) \frac{\delta^2(1 - g)(1 - g') + \delta(1 - g')}{1 - \delta^2(1 - g)(1 - g')}
    \;\;.
    \label{f-def:g-g:general}
\end{align}


\subsection{Proof of Proposition~\ref{prop:params}}
\label{app:gt:parameters}

Using the calculations from Section~\ref{app:gt:expected-payoffs:summary}
for the expected payoffs of $\gtft$ agents, we can now prove
Proposition~\ref{prop:params}, which is restated for convenience:

\params*

\begin{proof}
  Statements (ii) and (iii) follow directly from the definitions of
  $f(\G, \AC)$ and $f(\G, \AD)$
  in expressions~\eqref{f-def:g-c:general}
  and~\eqref{f-def:g-d:general}, respectively.
  For the former, observe that
  expression~\eqref{f-def:g-c:general} has no dependence on $g$,
  which proves statement (ii). 
  For the latter, observe that \eqref{f-def:g-d:general}
  is decreasing with $g$,
  which proves statement (iii).

  For statement (i), fix $g'' \in [0, \widehat g]$, and consider
  the definition of
  $f(\G, \G'')$ from
  expression~\eqref{f-def:g-g:general}.
  Differentiating this expression with respect to $g$ then finds
  \begin{equation}
    \begin{aligned}
      \frac{d}{d g}\; f(\G, \G'')
      &\;=\;
      (1-s_1)c \cdot \frac{(-\delta^2(1 - g'') - \delta)}{
        (1 - \delta^2(1 - g'')(1 - g))^2} \\
      &\qquad
        - (1-s_1) b \cdot \frac{(-\delta^2(1-g'') -\delta^3(1-g'')^2)}{
        (1 - \delta^2(1 - g)(1 - g''))^2} \;\;.
    \end{aligned}
    \label{eq:f-def:g-g:derivative}
  \end{equation}
  Then it is straightforward to check that when $\delta \geq \frac{c}{b}$
  when  $\widehat g < 1 - \frac{c}{\delta b}$,
  and when $s_1 \in [0,  1)$,
  (all of which are assumptions of the proposition),
  then expression~\eqref{eq:f-def:g-g:derivative} is 
  strictly positive for all $g, g'' \in [0, \widehat g]$.
  This implies under the mentioned constraints
  that for any $g, g'' \in [0, \widehat g]$, the
  function $f(\G, \G'')$
  is strictly increasing in $g$, which proves statement (i). 
\end{proof}

Observe that statement (ii) of the proposition
implies that only the expected payoff $f$
for nodes playing $\gtft$ strategy $g$
against $\AC$  is \textit{non-decreasing} with $g$,
whereas statements (i) and (iii) give
\textit{strictly increasing} inequalities with respect to $g$.
We remark that the $k$-IGT transitions could thus be adjusted
such that $\gtft$ agents only increase their parameter
following an interaction with a second $\gtft$ agent:
this would ensure \textit{strictly increasing}
expected payoff relationships for each transition type,
but at the expense of lower average stationary generosity values
(in the sense of Proposition~\ref{prop:avg-g},
and depending on the ratio of $\gtft$ to $\AC$ agents).

\section{Details on the average stationary generosity of $k$-IGT dynamics}
\label{app:asg}

In this section, we provide the
proof of Proposition~\ref{prop:avg-g},
which derives the average generosity parameter value
of $k$-IGT dynamics under its stationary distribution.
We restate the proposition for convenience.

\averagegen*

\begin{proof}
  Fix $k \ge 2$, and recall from Definition~\ref{def:kigt}
  that the $k$-IGT dynamics uses the set of generosity
  parameters $\calG = \{g_1,\dots, g_k\}$,
  where each $g_j = \frac{j-1}{k-1}\cdot \hg$.
  Then letting $\lambda := (1-\beta)/\beta$,
  Theorem~\ref{thm:kigt-main} implies
  (using properties of multinomial distributions) that
  \begin{equation*}
    \E\big[ \pi_j ]
    \;=\;
    m \cdot \frac{\lambda^{j-1} }{\sum_{i \in [k]} \lambda^{i-1}}
  \end{equation*}
  for all coordinates $j \in [k]$.
  It follows that the average stationary generosity $\wg$
  is given by
  \begin{equation*}
    \wg
    \;=\;
    \frac{1}{m}
    \sum_{j\in [k]}
    g_j \cdot \E[\pi_j]
    \;=\;
    \sum_{j\in [k]}
    \hg \cdot 
    \bigg(\frac{j-1}{k-1}\bigg)
    \cdot 
    \frac{\lambda^{j-1} }{\sum_{i \in [k]} \lambda^{i-1}}
    \;.
  \end{equation*}
  Now observe that when $\beta = \frac{1}{2}$ and 
  $\lambda = 1$, we have
  $\frac{\lambda^{j-1}}{\sum_{i\in[k]} \lambda^{i-1}} = \frac{1}{k}$
  for all $j \in [k]$, and thus it follows that $\wg = \hg/2$. 
  For the case when $\beta \neq 1/2$, we can simplify and write:
  \begin{align*}
    \wg
    \;=\;
      \sum_{j\in [k]}
      \hg \cdot 
      \bigg(\frac{j-1}{k-1}\bigg)
      \cdot 
      \frac{\lambda^{j-1} }{\sum_{i \in [k]} \lambda^{i-1}}
    &\;=\;
      \bigg( \frac{\hat{g}}{(k-1)\sum_{i=1}^k \lambda^{i-1}}\bigg)
      \cdot
      \sum_{j=1}^{k-1}
      j \lambda^j 
    \\
    &\;=\;
      \bigg( \frac{\hat{g}}{(k-1)\sum_{i=1}^k \lambda^{i-1}}\bigg)
      \cdot
      \sum_{i=1}^{k-1} \lambda^i \frac{\lambda^{k-i} - 1}{\lambda - 1}
    \\
    &\;=\;
      \bigg( \frac{\hat{g}}{(k-1)(\lambda - 1)\sum_{i=1}^k \lambda^{i-1}}\bigg)
      \cdot
      \sum_{i=1}^{k-1} (\lambda^{k} - \lambda^i)
      \\
    &\;=\;
      \bigg( \frac{\lambda \hat{g}}{(k-1)(\lambda - 1)\sum_{i=1}^k \lambda^{i-1}}\bigg)
      \cdot
      \bigg(
      \lambda^{k-1}(k-1) - \frac{\lambda^{k-1}-1}{\lambda - 1}
      \bigg) \;.
  \end{align*}
  Then using the fact that
  $\sum_{i\in[k]} \lambda^{i-1} = \frac{\lambda^{k}-1}{\lambda-1}$,
  we can further simplify and write
  \begin{align*}
    \wg
    &\;=\;
      \hg \cdot
      \bigg(
      \frac{\lambda}{(k-1)(\lambda^k - 1)}
      \bigg)
      \bigg(
      \lambda^{k-1} (k-1) - \frac{\lambda^{k-1}-1}{\lambda - 1}
      \bigg) \\
    &\;=\;
      \hg \cdot
      \bigg(
      \frac{\lambda^k}{\lambda^k - 1} -
      \bigg(\frac{1}{k-1}\bigg)
      \bigg(\frac{\lambda}{\lambda-1}\bigg)
      \bigg(\frac{\lambda^{k-1}-1}{\lambda^k -1}\bigg)
      \bigg) \;,
  \end{align*}
  which concludes the proof.
\end{proof}

As a corollary to Proposition~\ref{prop:avg-g}, we have the
following lower bound on $\wg$
in the regime where $\beta < 1/2$ (and thus $\lambda > 1$),
which shows that $\wg \approx \hg \cdot (1 - 1/k)$:

\begin{corollary}
  \label{corr:avg-g}
  Consider the setting of Proposition~\ref{prop:avg-g}
  for $k \ge 2$, and assume that $\beta < 1/2$, meaning
  $\lambda = \frac{(1-\beta)}{\beta} > 1$. Then
  $
  \wg \;\ge\;
  \hg \cdot
  \big(
  1 -
  \frac{1}{(\lambda - 1)(k-1)}
  \big).$
\end{corollary}

\begin{proof}
  Given that $\lambda > 1$, we have
  $\frac{\lambda^k}{\lambda^k-1} \ge 1$
  and
  $\frac{\lambda^{k-1} -1}{\lambda^k - 1}
  \le \frac{\lambda^{k-1}}{\lambda^k}$.
  It follows that
  \begin{equation*}
  \wg
  \;\ge\;
  \hg \cdot
  \bigg(
  1 -
  \Big(\frac{1}{k-1}\Big)
  \Big(\frac{\lambda}{\lambda - 1})
  \Big(\frac{\lambda^{k-1}}{\lambda^k}\Big)
  \bigg)
  \;=\;
  \hg \cdot
  \bigg(
  1 -
  \frac{1}{(\lambda - 1)(k-1)}
  \bigg) \;.
  \qedhere
\end{equation*}
\end{proof}


\section{Details on $k$-IGT Convergence to
  a Distributional Equilibrium}
\label{app:eq}

In this section we develop the proof of
Theorem~\ref{thm:kigt-eq-converge}, which shows
that the \textit{mean} stationary distribution of the
$k$-IGT dynamics is an $\eps$-approximate DE,
for $\eps = O(1/k)$.
Restated here:

\kigteqconvergence*

We also provide again the table summarizing 
the key RD game and population parameters:
\begin{table}[h!]
  \small
\begin{center}
  \begin{tabular}{cllcl}
    \toprule
    Symbol & Definition
    & & Symbol & Definition\\
    \midrule
    $b, c$ & donation game reward parameters
    & &
        $\alpha$ & fraction of $\AC$ agents
    \\
    $\delta$ & continuation probability
    & &
        $\beta$ & fraction of $\AD$ agents
    \\
    $s_1$ & initial cooperation probability
    & &
        $\gamma$ & fraction of $\gtft$ agents
    \\
    $\hatg$ & maximum generosity parameter
    & &
        $k$ & number of $\gtft$ parameter values \\
\bottomrule
\end{tabular}
\end{center}
\caption{Summary of Notation in RD Games
and $(\alpha,\beta, \gamma)$ populations}
\label{table:summary-small-repeat}
\end{table}

\subsection{Overview of the Proof}
We start by reviewing the strategy of the proof,
which expands on the high-level sketch
introduced in Section~\ref{sec:tech-overview:convergence}.
First, recall from the Definition~\ref{def:de} that
$\bfmu = \frac{1}{m}\E[\bfpi]$ is
an $\eps$-approximate DE when
\begin{equation}
  \E_{g \sim \bfmu, \;\S \sim \hatbfmu}
  \big[f(g, \S)\big]
  \;\ge\;
  \max_{g' \in \calG} \;
  \E_{\S \sim \hatbfmu}
  \big[f(g', \S)\big]
  - \eps
  \;.
  \label{eq:app-de:1}
\end{equation}
Here, recall that $\hatbfmu \in \Delta(\calS)$ is
the distribution induced by $\bfmu \in \Delta(\calG)$
in the $(\alpha, \beta, \gamma)$ population,
and $f$ is the expected payoff function
for the RD game, as derived in
Section \ref{app:gt:expected-payoffs}.
Given that $\calG = \{g_1, \dots, g_k\}$, 
the condition in expression~\eqref{eq:app-de:1}
can be equivalently written as
\begin{equation}
  \E_{\G \sim \bfmu, \;\S \sim \hatbfmu}
  \big[f(\G, \S)\big]
  \;\ge\;
  \max_{i \in [k]}\;
  \E_{\S \sim \hatbfmu}
  \big[f(\G_i, \S)\big]
  - \eps
  \;.
  \label{eq:app-de:2}
\end{equation}
Rearranging terms, the distribution $\bfmu$ must satisfy
\begin{equation}
  \max_{i \in [k]}\;
  \E_{\S \sim \hatbfmu}
  \big[f(\G_i, \S)\big]
  -
  \E_{\G \sim \bfmu, \;\S \sim \hatbfmu}
  \big[f(\G, \S)\big]
  \;\le\;
  \eps
  \;.
  \label{eq:app-de:3}
\end{equation}
For convenience, let $\Psi$ denote
the left hand side of expression~\eqref{eq:app-de:3},
meaning our goal in proving Theorem~\ref{thm:kigt-eq-converge}
is to show that $\Psi \le O(1/k)$. 
For this, we use a first-order
Taylor approximation argument to show that,
so long as the magnitude of the
second derivatives of $f(\cdot, \cdot)$
(with respect to its first argument)
are uniformly bounded by a constant $L > 0$:
\begin{equation}
  \E_{\G \sim \bfmu, \;\S \sim \hatbfmu} 
  \big[f(\G, \S)\big]
  \;\ge\;
  \E_{\S \sim \hatbfmu} 
  \big[f(\tildeG, \S)\big]
  \;-\;
  L \cdot \Var_{\G\sim \bfmu}[\G] \;,
  \label{eq:app-de:4}
\end{equation}
where $\tildeG := \E_{\G \sim \bfmu}[\G]$
is the average stationary generosity value
from Proposition~\ref{prop:avg-g}.
This inequality is established formally in
Proposition~\ref{prop:taylor-approx}.
Then it follows that we can further write
\begin{align}
  \Psi
  &\;:=\;
    \max_{i \in [k]}\;
    \E_{\S \sim \hatbfmu}
    \big[f(\G_i, \S)\big]
    -
    \E_{\G \sim \bfmu, \;\S \sim \hatbfmu}
    \big[f(\G, \S)\big]
    \nonumber \\
  &\;\le\;
    \max_{i \in [k]}\;
    \E_{\S \sim \hatbfmu}
    \big[f(\G_i, \S)\big]
    - 
    \E_{\S \sim \hatbfmu} 
    \big[f(\tildeG, \S)\big]
    + 
    L \cdot \Var_{\G\sim \bfmu}[\G]
  \nonumber \\
  &\;=\;
    \max_{i \in [k]}\;
    \E_{\S \sim \hatbfmu}
    \big[
    f(\G_i, \S)
    -
    f(\tildeG, \S)
    \big]
    + 
    L \cdot \Var_{\G\sim \bfmu}[\G]
    \;.
    \label{eq:app-de:5}
\end{align}
Under the parameter assumptions from the
statement of the theorem, we can bound the two terms in 
expression~\eqref{eq:app-de:5} separately as follows:
\begin{align*}
  &\;
    \text{Propositions~\ref{prop:var-G}
    and~\ref{prop:df2-bounded}}:
    \;\;\;
    L \cdot \Var_{\G\sim \bfmu}[\G] = O(1/k^2) \\
  &\;
    \text{Proposition~\ref{prop:gamma-bound}}:
    \;\;\;
    \max_{i \in [k]}\;
    \E_{\S \sim \hatbfmu}
    \big[
    f(\G_i, \S)
    -
    f(\tildeG, \S)
    \big] = O(1/k) \;.
\end{align*}
The statement of Theorem~\ref{thm:kigt-eq-converge}
then follows directly by combining these two results.
We proceed to establish these bounds separately
in the following subsections:

\subsection{Part I: Approximation via the
  Average Generosity $\tildeG$}

To start, we establish the claim from
expression~\eqref{eq:app-de:4}, where we approximate
$\E[f(\G, \S)]$ by $\E[f(\tildeG, \S)]$ and
pay a subsequent approximation error. Specifically,
we prove the following:

\begin{proposition}
  \label{prop:taylor-approx}
  Assume the conditions in Theorem~\ref{thm:kigt-eq-converge},
  and let $\tildeG := \E_{\G \sim \bfmu}[\G]$.
  Moreover, assume that $|f''(g, \cdot)| \le L$ uniformly,
  where the second derivatives are taken with respect to $g$.
  Then
  \begin{equation*}
  \E_{\G \sim \bfmu, \;\S \sim \hatbfmu} 
  \big[f(\G, \S)\big]
  \;\ge\;
  \E_{\S \sim \hatbfmu} 
  \big[f(\tildeG, \S)\big]
  \;-\;
  L \cdot \Var_{\G\sim \bfmu}[\G] \;.
\end{equation*}
\end{proposition}

\begin{proof}
  Assuming the conditions of the proposition,
  it suffices to prove that 
  \begin{equation*}
    \Big|\;
    \E_{\G \sim \bfmu, \;\S \sim \hatbfmu} 
    \big[f(\G, \S)\big]
    \;-\;
    \E_{\S \sim \hatbfmu} 
    \big[f(\tildeG, \S)\big]
    \;\Big|
    \;\le\;
    L \cdot \Var_{\G\sim \bfmu}[\G] \;.
  \end{equation*}
  For this, let us first define the function
  $F(\G) := \E_{\S \sim \hatbfmu}
  \big[f(\G, \S)\big]$.
  Then taking a first-order Taylor approximation
  with remainder of $F$ centered at $\tildeG \in [0, \hatg]$,
  it follows that
  \begin{equation*}
    F(\G)
    \;=\;
    F(\tildeG)
    + F'(\tildeG)(\G- \tildeG)
    + F''(\xi) (\G - \tildeG)^2 \;,
  \end{equation*}
  for some $\xi$ between $\tildeG$ and $\G$.
  Taking expectation over $\G \sim \bfmu$, we find:
  \begin{align*}
    \E_{\G \sim \bfmu}
    \big[F(\G)\big]
    &\;=\;
      F(\tildeG)
      +
      F(\tildeG) \cdot 
      \E_{\G \sim \bfmu}\big[\G- \tildeG\big]
      + F''(\xi) \cdot
      \E_{\G \sim \bfmu}\big[(\G- \tildeG)^2\big] \\
    &\;=\;
      F(\tildeG)
      +
      F''(\xi)
      \cdot
      \Var_{\G \sim \bfmu}[\G] \;.
  \end{align*}
  Here, the first equality follows from the
  fact that $F(\tildeG)$ is a constant, and the
  second line follows by the definition of variance
  and from observing that $\E[\G - \tildeG] = 0$. 
  Rearranging terms and taking absolute values, we find
  \begin{equation*}
    \Big|
    \E_{\G \sim \bfmu}
    \big[F(\G)\big]
    - 
    F(\tildeG)
    \Big|
    \;\le\;
    F''(\xi) \cdot
    \Var_{\G \sim \bfmu}[\G] \;.
  \end{equation*}
  Now by definition, $F(\xi)$ is an average
  of $f(\xi, \S)$ for $\S \sim \hatbfmu$.
  Thus $F''(\xi)$ is an average over $f''(\xi, \S)$,
  and by the assumptions of the lemma,
  all $|f''(\xi, \S)|$ are bounded above by $L$ uniformly. 
  It follows that 
  \begin{equation*}
    \Big|
    \E_{g \sim \bfmu}
    \big[F(g)\big]
    - 
    F(\tildeg)
    \Big|
    \;\le\;
    L \cdot
    \Var_{\G \sim \bfmu}[\G] \;,
  \end{equation*}
  which completes the proof.
\end{proof}

In the following two propositions, we derive
upper bounds on the constant $L$
and the variance of $\G \sim \bfmu$ used in
Proposition~\ref{prop:taylor-approx}.

\begin{proposition}
  \label{prop:var-G}
  Assume the conditions of Theorem~\ref{thm:kigt-eq-converge}.
  Then $\Var_{\G \sim \bfmu}[\G] \le \frac{16}{(k-1)^2}$. 
\end{proposition}

\begin{proof}
  First, recall that $\G \in \{g_1, \dots, g_k\}$,
  where $g_i := \hatg \big(\frac{i-1}{k-1}\big)$
  for each $i \in [k]$, 
  and $\tildeG = \E_{g \sim \bfmu}[g]$.
  Moreover, recall that
  $\bfmu(i) \propto \lambda^{i-1}$ for each $i \in [k]$.
  Then using the definition of variance
  $\Var_{g \sim \bfmu} = 
  \E_{\G \sim \bfmu}
  \big[\G^2\big]
  - \tildeG^2$, it follows that we can write
  \begin{align}
    \Var_{\G \sim \bfmu}\big[\G\big]
    &\;=\;
      \sum_{i=1}^k
      \bfmu(i) \cdot \hatg^2
      \cdot \Big(\frac{i-1}{k-1} \Big)^2
      -
      \bigg(
      \sum_{i=1}^k
      \bfmu(i) \cdot \hatg
      \cdot \Big(\frac{i-1}{k-1}\Big)
      \bigg)^2
      \nonumber \\
    &\;\le\;
      \sum_{i=1}^k
      \frac{\lambda^{i-1}}{\sum_{j=1}^k \lambda^{i-1}}
      \cdot \Big(\frac{i-1}{k-1} \Big)^2
      -
      \bigg(
      \sum_{i=1}^k
      \frac{\lambda^{i-1}}{\sum_{j=1}^k \lambda^{i-1}}
      \cdot \Big(\frac{i-1}{k-1}\Big)
      \bigg)^2
      \label{eq:var:1}
      \;,
  \end{align}
  where the inequality comes from factoring
  out $\hatg^2 \le 1$.
  Then by a calculation similar to that in
  the proof of Proposition~\ref{prop:avg-g}, we 
  can simplify expression~\eqref{eq:var:1} to 
  further write
  \begin{align}
    \Var_{\G \sim \bfmu}\big[\G\big]
    &\;\leq\;
      \frac{
      2(k^2-1)\lambda^{k+1}
      + \lambda(\lambda^{2k} + 1)
      - k^2 \lambda^k(1+ \lambda^2)
      }
      {(k-1)^2 (\lambda -1)^2 (\lambda^k -1)^2}
      \nonumber \\
    &\;\le\;
      \frac{2(k^2-1)\lambda^{k+1}}
      {(k-1)^2 (\lambda -1)^2 (\lambda^k -1)^2}
      +
      \frac{\lambda(\lambda^{2k} + 1)}
      {(k-1)^2 (\lambda -1)^2 (\lambda^k -1)^2}
      \;,
      \label{eq:varG:1}
  \end{align}
  where the inequality follows from the fact
  that $k^2\lambda^k(1+\lambda^2) \ge 0$ for all
  $\lambda, k \ge 0$.
  Thus to prove the statement of the proposition,
  it suffices to bound each term in
  expression~\eqref{eq:varG:1} separately.

  For the first term in~\eqref{eq:varG:1}, observe that
  $\frac{2(k^2-1)}{(k-1)^2}$ is decreasing with $k$,
  and is bounded above by 6 for all $k\ge 2$.
  Moreover, $\frac{\lambda^{k+1}}{(\lambda-1)^2(\lambda^k-1)^2}$
  is also decreasing with $k$ and $\lambda$ and is bounded
  above by $\frac{2}{(k-1)^2}$ for $k\ge 2$ and $\lambda \ge 2$.
  Together, this implies
  \begin{equation}
    \frac{2(k^2-1)\lambda^{k+1}}
    {(k-1)^2 (\lambda -1)^2 (\lambda^k -1)^2}
    \;\le\;
    \frac{12}{(k-1)^2}
    \label{eq:varG:2}
  \end{equation}
  for all $k, \lambda \ge 2$.
  For the second term in~\eqref{eq:varG:1}, observe that
  $\frac{\lambda(\lambda^{2k} + 1)}{(\lambda-1)^2(\lambda^k-1)^2}$
  is decreasing with $k$ and $\lambda$ and is bounded
  above by 4 for $k, \lambda \ge 2$. Thus
  \begin{equation}
    \frac{\lambda(\lambda^{2k} + 1)}
    {(k-1)^2 (\lambda -1)^2 (\lambda^k -1)^2}
    \;\le\;
    \frac{4}{(k-1)^2}
    \label{eq:varG:3}
  \end{equation}
  for all $k, \lambda \ge 2$.
  Substituting the bounds from expressions
  \eqref{eq:varG:2} and~\eqref{eq:varG:3} into
  expression~\eqref{eq:varG:1} concludes the proof.
\end{proof}

Finally, in the following proposition, we derive a
uniform, constant bound on the magnitude
of the second derivatives of $f(g, \cdot)$
with respect to $g$:

\begin{proposition}
    \label{prop:df2-bounded}
    Assume the conditions of Theorem~\ref{thm:kigt-eq-converge}.
    Then for every $g \in [0,\hatg]$ and
    $\S \in \calS = \{\AC, \AD, g_1, \dots, g_k\}$, the magnitude
    of the second derivatives of $f(g, \S)$ with respect to $g$
    are bounded by an absolute constant $L > 0$. 
\end{proposition}

\begin{proof}
  Using the expressions for $f(g, \cdot)$
  from Section~\ref{app:gt:expected-payoffs:summary},
  observe that
  $\frac{d^2}{dg^2} f(g, \AC)$
  and $\frac{d^2}{dg^2} f(g, \AD)$ are both 0,
  and thus it suffices to show that
  $\big|\frac{d^2}{dg^2} f(g, g')\big|$ is bounded above
  by an absolute constant for all $g, g' \in [0, \hatg]$.

  For this, we differentiate the
  calculation of $\frac{d}{dg} f(g, g')$ from
  expression \eqref{eq:f-def:g-g:derivative} to find:
  \begin{equation}
    \begin{aligned}
      \frac{d^2}{d g^2}\; f(g, g')
      \;=\;
      (1-s_1) \cdot \bigg(
      c \cdot \frac{2\delta^3(1-g')(1+\delta(1-g'))}{
      (1 - \delta^2(1 - g')(1 - g))^3}
      - b \cdot \frac{2\delta^4(1-g')^2(1+\delta(1-g'))}{
      (1 - \delta^2(1 - g)(1 - g'))^3} \bigg)\;\;.
    \end{aligned}
  \end{equation}
  It follows that
  \begin{equation}
    \frac{d^2}{d g^2}\; f(g, g')
      \;\le\;
      c \cdot (1-s_1) \cdot \bigg(
      \frac{2\delta^3(1-g')(1+\delta(1-g'))}{
        (1 - \delta^2(1 - g')(1 - g))^3}\bigg)
      \label{eq:d2g:1}
  \end{equation}
  and also that
  \begin{equation}
    \frac{d^2}{d g^2}\; f(g, g')
    \;\ge\;
    - b \cdot (1-s_1) \cdot \bigg(
    \frac{2\delta^4(1-g')^2(1+\delta(1-g'))}{
      (1 - \delta^2(1 - g)(1 - g'))^3} \bigg) \;.
    \label{eq:d2g:2}
  \end{equation}
  Note that $s_1, c, b, \delta$ and $\hatg$
  are all absolute constants. Thus it follows
  from expressions~\eqref{eq:d2g:1} and~\eqref{eq:d2g:2}
  that there exists an absolute constant $L> 0$ such that
  $\big|\frac{d^2}{d g^2}f(g, g') \big| \le L$
  for all $g, g' \in [0, \hatg]$, which concludes the proof.
\end{proof}

\subsection{Part II: Upper Bounds on the Approximation}

In the next proposition, we derive an upper bound on difference
$
\max_{i \in [k]}\;
\E_{\S \sim \hatbfmu}
\big[
f(\G_i, \S)
-
f(\tildeG, \S)
\big]
$
from the first term in expression~\eqref{eq:app-de:5}.

\begin{proposition}
  \label{prop:gamma-bound}
  Assume the conditions of Theorem~\ref{thm:kigt-eq-converge}.
  Specifically, consider RD settings
  with $s_1 \in [0, 1)$,
  $\frac{b}{c} > 1+ \frac{\beta c}{\gamma (1-s_1)}$,
  $\delta < \sqrt{1-\frac{\beta c}{\gamma(b-c)(1-s_1)}}$,
  and 
  $
  \hatg
  \;<\;
  1-\frac{1}{\delta}
  \Big(\frac{\beta c}{\gamma (b-c)(1-\delta)(1-s_1)} - 1\Big).
  $
  Then
  \begin{equation*}
    \max_{i \in [k]}\;
    \E_{\S \sim \hatbfmu}
    \big[
    f(\G_i, \S)
    -
    f(\tildeG, \S)
    \big]
    \;\le\;
    O\Big(\frac{1}{k}\Big) \;,
  \end{equation*}
  where the constant in the $O(\cdot)$ depends only
  on $b, c, \gamma, \beta,$ and $\delta$.
\end{proposition}

\begin{proof}
  For convenience, let
  $
  \Gamma(i)
  :=
  \E_{\S \sim \hatbfmu}
  \big[f(\G_i, \S) - f(\tildeG, \S)\big]
  $
  for each $i \in [k]$.
  Then by definition, we have
  \begin{align}
    \Gamma(i)
    &\;=\;
      \alpha
      \big(
      f(\G_i, \AC) - f(\tildeG, \AC)
      \big)
      +
      \beta
      \big(
      f(\G_i, \AD) - f(\tildeG, \AD)
      \big)
      +
      \gamma
      \Big(
      \sum_{j=1}^k
      \bfmu(j)
      \big(
      f(\G_i, \G_j) - f(\tildeG, \G_j)
      \big)
      \Big)
      \nonumber \\
    &\;=\;
      \frac{\beta \delta c (\tildeg - g_i)}{(1-\delta)}
      +
      \gamma
      \Big(
      \sum_{j=1}^k
      \bfmu(j)
      \big(
      f(\G_i, \G_j) - f(\tildeG, \G_j)
      \big)
      \Big)
      \;,
      \label{eq:gamma-i:1}
  \end{align}
  where the equality comes from applying the
  expressions for $f(\cdot, \AC)$
  and $f(\cdot, \AD)$ from Section~\ref{app:gt:expected-payoffs:summary}.
  Now recall that $\bfmu(j) \propto \lambda^{j-1}$
  for each $j \in [k]$, and for $\lambda \ge 2$,
  it follows that
  \begin{equation}
    \bfmu(k)
    \;=\;
    1 - \sum_{i=1}^k \bfmu(i)
    \;=\;
    1 - 
    \frac{\sum_{i=1}^{k-1}\lambda^{i-1}}
    {\sum_{j=1}^k \lambda^{j-1}}
    \ge
    1 - \frac{1}{k}
    \label{eq:mu-k-lb}
  \end{equation}
  for all $k \ge 2$. Together with the fact that
  $f(\cdot, \G_j)$ is uniformly bounded by an absolute
  constant $M > 0$ for all $j \in [k]$, we find
  \begin{align*}
    \sum_{j=1}^k
    \bfmu(j)
    \big(
    f(\G_i, \G_j) - f(\tildeG, \G_j)
    \big)
    &\;\le\;
      \Big(\sum_{j=1}^{k-1} M\cdot\bfmu(j) \Big)
      +
      \bfmu(k)
      \big(f(\G_i, \G_k) - f(\tildeG, \G_k)\big) \\
    &\;=\;
      M\cdot(1- \bfmu(k))
      +
      \bfmu(k)
      \big(f(\G_i, \G_k) - f(\tildeG, \G_k)\big) \\
    &\;\le\;
      O\Big(\frac{1}{k}\Big)
      +
      \big(f(\G_i, \G_k) - f(\tildeG, \G_k)\big) \;,
  \end{align*}
  where in the final line we applied the lower bound on
  $\bfmu(k)$ from expression~\eqref{eq:mu-k-lb} and 
  used the fact that $\bfmu(k) \le 1$.
  Substituting this back into expression~\eqref{eq:gamma-i:1},
  we can write
  \begin{equation}
    \Gamma(i)
    \;\le\;
    \frac{\beta \delta c(\tildeg - g_i)}{(1-\delta)}
    + \gamma
    \big(f(\G_i, \G_k) - f(\tildeG, \G_k)\big)
    + O\Big(\frac{1}{k}\Big)
    \;.
    \label{eq:gamma-i:2}
  \end{equation}
  Notice that to prove the statement of the proposition,
  our goal is to  show that
  $\max_{i \in [k]} \Gamma(i) \le O(1/k)$,
  and for this, observe that it now suffices to upper bound the 
  first two terms of expression~\eqref{eq:gamma-i:2} by $O(1/k)$. 
  For this, using the expected payoff expressions for 
  $f(\cdot, \cdot)$ from Section~\ref{app:gt:expected-payoffs:summary},
  a straightforward calculation finds that
  \begin{align}
    f(\G_i, \G_k) - f(\tildeG, \G_k)
    &\;=\;
      (1-s_1) \cdot 
      \big(
      \delta c(\tildeg - g_i)
      +
      \delta^2 (b-c)
      (g_i - \tildeg)(1-\hatg)
      +
      \delta b
      (g_i - \tildeg)
      \big)
      \cdot \frac{1}{\Phi}
      \nonumber \\
    &\;=\;
      (g_i - \tildeg)
      \cdot
      (1-s_1)
      \cdot
      \big(
      (b-c)
      (\delta^2 (1-\hatg) + \delta)
      \big)
      \cdot
      \frac{1}{\Phi}
      \label{eq:gamma-i:3}
    \\
    \text{where}\;\;
    \Phi
    &\;:=\;
      (1-\delta^2(1-\tildeg)(1-\hatg))
      \cdot
      (1-\delta^2(1-g_i)(1-\hatg))
      \;.
      \label{eq:gamma-i:4}
  \end{align}
  Substituting this back into
  expression~\eqref{eq:gamma-i:2} yields
  \begin{equation}
    \Gamma(i)
    \;\le\;
    (g_i - \tildeg)
    \cdot
    \Big(
    \gamma
    (1-s_1)\Big(
    \frac{(b-c)}{2}
    \cdot
    \big(\delta^2 (1-\hatg) + \delta\big)
    \Big)
    \cdot
    \frac{1}{\Phi}
    -
    \frac{\beta \delta c}{(1-\delta)}
    \Big)
    + O\Big(\frac{1}{k}\Big)
    \;.
    \label{eq:gamma-i:5}
  \end{equation}
  Now using the fact that $\Phi \le 1$ uniformly
  for all $i \in [k]$, it follows that
  \begin{equation*}
    \gamma
    (1-s_1)\big(
    (b-c)
    \cdot
    \big(\delta^2 (1-\hatg) + \delta\big)
    \big)
    \cdot
    \frac{1}{\Phi}
    -
    \frac{\beta \delta c}{(1-\delta)}
    \;>\;
    0
  \end{equation*}
  so long as
  \begin{equation*}
    \hatg
    \;<\;
    1-\frac{1}{\delta}
    \bigg(\frac{\beta c}{(b-c)\gamma(1-\delta)(1-s_1)} - 1\bigg)
    \;,
  \end{equation*}
  which holds by assumption in the statement of the proposition. 
  Note that the additional assumptions that
  $\delta < \sqrt{1-\frac{\beta c}{\gamma (b-c)(1-s_1)}}$
  and $\frac{b}{c} > 1+ \frac{\beta c}{\gamma (1-s_1)}$
  ensure that (a) the upper bound constraint on $\hatg$
  is positive, and that (b) the subsequent
  upper bound constraint on $\delta$ is positive.

  It follows that expression~\eqref{eq:gamma-i:5} is
  maximized by the largest $j \in [k]$ such that
  $g_j \ge \tildeg$. For this, recall that
  $g_k = \hatg$ by definition and that
  $\tildeg \ge \hatg \big(1 - \frac{1}{(\lambda-1)(k-1)}\big)$
  by Corollary~\ref{corr:avg-g},
  and together this implies that
  $(g_i - \tildeg) \le \frac{\hatg}{(\lambda-1)(k-1)}$
  for all $i \in [k]$.
  Combined with the fact that $1/\Phi \le 1/(1-\delta^2)^2$,
  we find 
  \begin{equation*}
    \Gamma(i)
    \;\le\;
    \frac{1}{(\lambda-1)(k-1)}
    \cdot
    \frac{(1-s_1)}{(1-\delta^2)^2}\Big(
    \frac{(b-c)}{2}
    \cdot
    \big(\delta^2 + \delta\big)
    \Big)
    + O\Big(\frac{1}{k}\Big) \;,
  \end{equation*}
  where we used the additional fact that $\hatg \le 1$.
  Given that $\lambda \ge 2, b, c,\delta,$ and $s_1$
  are all absolute constants, this shows
  $\Gamma(i) \le O(1/k)$ for all $i \in [k]$, which 
  concludes the proof.
\end{proof}

\end{document}